\documentclass[preprint]{imsart}
\RequirePackage[OT1]{fontenc}
\RequirePackage{amsthm,amsmath}
\RequirePackage[numbers]{natbib}
\RequirePackage[colorlinks,citecolor=blue,urlcolor=blue]{hyperref}

\issue{}
\firstpage{}
\lastpage{}

\startlocaldefs
\usepackage[tikz]{bclogo}
\usepackage{algorithm2e}
\usepackage[normalem]{ulem}

\usepackage{amssymb, amsmath, natbib, amsthm, graphicx, euscript, mathrsfs, enumerate }
\def\sec#1{\underline{\textbf{#1}}}

\def\R{I\!\!R}
\def\I{I\!\!I}
\renewcommand{\i}{\mathrm{i}}

\newcommand{\eqd}{\stackrel{d}{=}}

\renewcommand{\P}{{\mathbb P}}

\newcommand{\N}{{\mathbb N}}

\newcommand{\T}{\mathcal{T}}
\newcommand{\E}{{\mathbb E}}

\renewcommand{\L}{{\EuScript L}}
\newcommand{\ii}{\i}
\newcommand{\eps}{\varepsilon}
\newcommand{\Var}{\operatorname{Var}}
\renewcommand{\Re}{\mathrm{Re}}
\renewcommand{\Im}{\mathrm{Im}}
\newcommand{\F}{\mathcal{F}}
\newcommand{\M}{\mathcal{M}}
\theoremstyle{plain}
\newtheorem{thm}{Theorem}[section]

\newtheorem{rem}[thm]{Remark}
\newtheorem{prop}[thm]{Proposition}
\theoremstyle{definition}
\newtheorem{defi}[thm]{Definiton}
\newtheorem{ex}[thm]{Example}

\newcommand{\argmin}{\operatornamewithlimits{arg\,min}}

\newcommand{\C}{{\mathbb C}}
\newcommand{\vsp}{\vspace{0.5cm}}
\newcommand{\K}{\mathcal{K}}
\newcommand\erf{\operatorname{erf}}

\renewcommand{\c}{\mu}
\renewcommand{\a}{\lambda}
\newcommand{\uc}{u^{\circ}}

\renewcommand{\C}{R}

\newcommand{\mE}{\mathcal{E}}
\newcommand{\mg}{\mathcal{G}}

\begin{document}
\begin{frontmatter}
\title{Statistical inference for generalized  Ornstein-Uhlenbeck processes\thanksref{T1}}
\runtitle{Statistical inference for GOU}

\begin{aug}
\author{\fnms{Denis} \snm{Belomestny}\ead[label=e1]{denis.belomestny@uni-due.de}}
\address{University of Duisburg-Essen\\
Thea-Leymann-Str. 9, 45127 Essen,  Germany\\
and \\
National Research University Higher School of Economics \\ Shabolovka, 26, Moscow,  119049 Russia\\
\printead{e1}}

\author{\fnms{Vladimir} \snm{Panov}\ead[label=e2]{vpanov@hse.ru}}
\address{National Research University Higher School of Economics \\ Shabolovka, 26, Moscow,  119049 Russia.\\
\printead{e2}}

\thankstext{T1}{
The financial support from the Government of the Russian Federation within the framework of the implementation of the 5-100 Programme Roadmap of the National Research University  Higher School of Economics is acknowledged. 
}
\runauthor{D.Belomestny and V.Panov}

\end{aug}

\begin{abstract}
In this paper, we consider the problem of statistical inference for generalized Ornstein-Uhlenbeck processes of the type 
\[
X_{t} = e^{-\xi_{t}} \left( X_{0} + \int_{0}^{t} e^{\xi_{u-}} d u \right), 
\]
where    \(\xi_s\) is a  L{\'e}vy process. Our primal goal is to estimate the characteristics of the L\'evy process \(\xi\) from the low-frequency  observations of the process \(X\). We present a novel approach towards estimating the L{\'e}vy triplet of \(\xi,\) which is based on the Mellin transform technique. It is shown that the resulting estimates attain optimal minimax convergence rates.  The suggested  algorithms are illustrated by numerical simulations. 
\end{abstract}

\tableofcontents
\end{frontmatter}

\section{Introduction}
Let \((\xi_t)_{t\geq 0}\) be a L\'evy process with a L{\'e}vy triplet \(\left( \mu, \sigma^{2}, \nu \right)\). The main object of our study is the so-called generalized Ornstein-Uhlenbeck (GOU) process defined as
\begin{eqnarray}
\label{GOU}
X_{t} = e^{-\xi_{t}} \left(  
	X_{0} + \int_{0}^{t} e^{\xi_{u-}} d u
\right),\quad t\geq 0.
\end{eqnarray}
The GOU processes have recently got much attention in the literature. A comprehensive study of the GOU processes and an extended list of references  can be found in the theses of Behme \cite{Behme}, where, in particular, it is shown that \(X_{t}\) satisfies the following SDE: 
\begin{eqnarray*}
d X_{t}= X_{t-} dU_{t} +dt, \quad \mbox{where} \quad 
U_{t}:= -\xi_{t} +\sum_{0<s\leq t} \left( 
	e^{-\Delta \xi_{s}} - 1 + \Delta \xi_{s}
\right) +\frac{1}{2} \sigma^{2} t.
\end{eqnarray*}
The popularity of GOU processes is related to the fact they appear to be useful in several applications.
For instance, the process \eqref{GOU}  determines the volatility process in the COGARCH (COntinious Generalized AutoRegressive Conditionally Heteroscedastic) model  introduced in Kl{\"u}ppelberg et al. \cite{cogarch}.
One important result from the theory of GOU processes is that, under some conditions, the process \eqref{GOU} is 
stationary with invariant stationary distribution given by the distribution of the following exponential functional of \(\xi:\)
\begin{eqnarray}
\label{Ainfty}
A_{\infty}:= \int_{0}^{\infty} e^{-\xi_{t}} \; dt.
\end{eqnarray}  
In fact, the properties of the functional \eqref{Ainfty} have been widely studied in the literature and we refer to the survey by Bertoin and Yor \cite{BertoinYor} for a theoretical background of the exponential functionals. In particular, it is known that the Mellin transform of the density \(\pi\) of exponential functional, 
\[\M(z):= \E \left[A_{\infty}^{z-1}\right]=\int_{0}^{\infty}x^{z-1}\pi(x)\, dx,\]  
satisfies the following recursive formula
\begin{eqnarray}
\label{momenty}
	\M(z)
=
\frac{
	\phi(z)
}{
	z
}\;
\M(z+1),
\end{eqnarray}
where \(\phi(z)\) is a Laplace exponent of the process \(\xi\), i.e., 
\(
	\phi(z) := - \log \E \left[
		e^{ - z \xi_{1}}
	\right],
\)
and complex \(z\) is taken from the strip
 \begin{eqnarray}
 \label{Ups}
 \Upsilon:= \Bigl\{z \in \C: \;  0<\Re(z)  < \theta \Bigr\} \quad \mbox{with}\quad  \theta:=\sup\left\{x \geq 0: \E[ e^{-x \xi_{1}}] \leq 1\right\}.
 \end{eqnarray}
The recursive  formula \eqref{momenty} first appeared  for real \(z\) in the paper by Maulik and Zwart \cite{MZ}. The validity of \eqref{momenty} for complex \(z\) was recently  shown by Kuznetsov, Pardo and Savov \cite{Kuznets}. If \(\xi_{t}\) is a subordinator,  the parameter \(\theta\) is equal to infinity. Let us note that 
the functional \(A_{\infty}\) appeared in such application areas as finance  (see, e.g. the monograph by Yor \cite{Yor}),    carousel systems (see Litvak and Adan \cite{LitvakAdan}, Litvak and van Zwet \cite{LitvakZwet}),  self-similar fragmentations (see Bertoin and Yor \cite{BertoinYor}), and information transmission problems (especially TCP/IP protocol, see Guillemin, Robert and Zwart \cite{Guill}).   For the detailed discussion on the physical interpretations, we refer to Comtet, Monthus and Yor \cite{CMY} and the dissertation by  Monthus \cite{Monthus}.
\par
In this paper, we mainly focus on the case when  \(\xi\) is a subordinator with finite L{\'e}vy measure. In terms of the L{\'e}vy triplet \((\mu,\sigma^2,\nu)\), this means that \(\mu>0\), \(\sigma=0\), \(\nu(\R_{-})=0\) and moreover \(\a:=\nu(\R_{+})<\infty\).  Suppose that the process \eqref{GOU} is observed at equidistant time points \(0=t_{0} < t_{1} < \ldots < t_{n}\). Since under some mild assumptions  the process is stationary and the invariant distribution is given by the distribution of the exponential functional \(A_{\infty}\) (see Fasen, \cite{Fasen}), we assume that \(X_{t_{0}}, \ldots , X_{t_{n}}\) are also distributed as  \(A_{\infty}\).  Our main goal  is statistical inference on the L{\'e}vy triplet \((\mu,\sigma^2,\nu)\) based on the observations  \(X_{t_{0}},  \ldots , X_{t_{n}}\).   More precisely, we will pursue the following two aims: (1) estimation of the drift term \(\c\)  and the intensity parameter \(\a\); (2)  estimation of the L{\'e}vy measure \(\nu\). 
\par
To the best of our knowledge, the statistical inference for GOU processes of the form \eqref{GOU} from their low-frequency observations has not been yet studied in  the literature. In fact the resulting statistical problem is quite challenging and needs careful treatment. Indeed, the only connection between the stationary distribution of a GOU process, which can be estimated from the data, and the parameters of the underlying L\'evy process is given by the recurrent relation \eqref{momenty} which is rather implicit.  The main idea of our procedure for estimating   the parameters of the process \(\xi\)  can be described as follows. First, by making use of \eqref{momenty}, we estimate the Laplace exponent \(\phi(z)\) at the points \(z=u^{\circ}+\ii v \in \Upsilon\), where \(u^{\circ}>0\) is fixed and \(v\) varies on the equidistant grid between \(\eps V_{n}\) and \(V_{n}\) (with \(\eps>0\) and \(V_{n} \to \infty\) as \(n \to \infty\)) Afterwards, we use the representation  
\begin{eqnarray}
	\label{estproc1}
	\phi(\uc+\ii v) &=& \a+ \c \left( \uc+\ii v \right) -  \F[\bar\nu](-v) , \qquad v \in \R,
\end{eqnarray}
where  \(\bar\nu(dx):= e^{-\uc x} \nu(dx)\), and \(\F[\bar\nu] (v)\) stands for the Fourier transform of the measure \(\bar\nu\), i.e., \( \F[\bar\nu](v) := \int_{\R_{+}} e^{\ii v x} \bar\nu(dx).\) Since \(\F[\bar\nu] (v) \to 0\) as \(v \to \infty\) by the Riemann-Lebesgue lemma, upon taking  real and imaginary parts of the left and right hand sides of \eqref{estproc1}, we are able to consequently estimate the parameters \(\c\) and \(\a\). With no doubt, the second aim, a complete recovering of the L{\'e}vy measure \(\nu,\) is the most difficult task. Since the estimates of the parameters \(\c\) and \(\a\) are already obtained, we can estimate by \eqref{estproc1} the Fourier transform \(\F[\bar\nu] (v)\) of \(\bar\nu \) for \(v\)  from  \([-V_{n}, V_{n}]\).   The last step of this procedure, the estimation of the L{\'e}vy measure \(\nu\),  is based on the regularised inverse Fourier transform formula. 

The above estimation algorithm bears some similarity to the spectral estimation algorithm introduced by Belomestny and Reiss \cite{DBReiss}, \cite{BelReiss}. Let us also mention that the problem of statistical inference for L\'evy processes (or some their generalizations)  observed at low frequency was the subject of many studies, see, e.g. Neumann and Rei{\ss}~\cite{NeuReiss},  
Rei{\ss}~\cite{TestingReiss}, Kappus~\cite{AdaptiveKappus}, Trabs~\cite{CalibrationTrabs} and Jongbloed et al. ~\cite{JMV}.
Note that the last reference deals with the L\'evy-driven Ornstein-Uhlenbeck processes, which are not of the form \eqref{GOU}.
\par
The paper is organized as follows. In the next section, we formulate our main assumptions  and give some examples. In Section~\ref{estimation},  the main estimation algorithm is presented and discussed in details. Next, we analyze the convergence rates of the proposed algorithms in Section~\ref{sec4} and provide some numerical examples in Section~\ref{secsim}. The proofs of our theoretical results are collected in Section~\ref{theory}.

\section{Main setup}
\label{main_setup}
 In this article, we study the class of subordinators with finite L{\'e}vy measures as possible choice for the L\'evy process $(\xi_{t})$. In terms of the L{\'e}vy triplet \((\c,\a,\nu)\), this means that
 \begin{eqnarray}
   \label{cond1}
    \left\{
    \begin{aligned}
    \c\geq 0, &\qquad \sigma=0, \\
    \nu(\R_{-}) = 0, &\qquad \a:=  \nu(\R_{+}) < \infty.\\
    \end{aligned}
    \right.
    \end{eqnarray}
A detailed discussion of the subordination theory as well as various examples of such processes (Gamma, Poisson,  tempered stable, inverse Gaussian, Meixner processes, etc.), are given in  \cite{BNS}, \cite{Bertoin}, \cite{ContTankov},  \cite{Sato}, \cite{Schoutens}. Note that in the case of subordinators, the truncation function in the L{\'e}vy-Khinchine formula can be omitted, that is,  the characteristic exponent of \(\xi\) is equal to 
\begin{eqnarray}
\label{phis}
	\psi(z) =  \log \E \left[
		e^{\ii z \xi_{1}}
	\right] 
	=
	 \ii \c z  
	 + \int_{0}^{\infty} 
	\left(
		e^{\ii z x}  - 1
	\right) \nu(dx).
\end{eqnarray}      
Later on, we also need  the  Laplace exponent of \(\xi\), which is defined as \[\phi(z) := - \log \E \left[
		e^{-z \xi_{1}}
	\right] =  -\psi\left(\ii z \right).\] Under the assumption \eqref{cond1}  the Laplace exponent \(\phi(\cdot)\) is given by
\begin{eqnarray}
\label{phis2}
	\phi(z) &=& 
	\c z 
	+ \int_{0}^{\infty}
	\left(
		1 -  e^{- z u} 
	\right) \nu(du) \\
\label{phis3}
&=& 
	\c z 
	+ z \int_{0}^{\infty} 
	  e^{- z u} \nu\left( (u, +\infty) \right) du.
\end{eqnarray}
Let us summarise the main properties of the functional $A_{\infty}=\int_{0}^{\infty}e^{-\xi_{t}}dt$
in this case. 
\begin{prop}
The random variable $A_{\infty}$ admits a bounded density $\pi$ and fulfills
$\mathrm{E}\left[A_{\infty}^{s-1}\right]<\infty$ for all $s>0.$
If $\mu>0$, then $0<A_{\infty}\leq1/\mu$ a.s. Moreover, the following
relation holds for $\mathrm{Re}[z]>0,$
\begin{eqnarray}
\nonumber
\phi(z)=
z\frac{\M(z)}{\M(z+1)} & = & \mu z+\int_{0}^{\infty}(1-e^{-zx})\nu(dx),\\
\label{Mz}
 & = & \lambda+\mu z-\int_{0}^{\infty}e^{-zx}\nu(dx),
\end{eqnarray}
where \(\M(z)\) is the Mellin transform of \(\pi.\)
\end{prop}

\section{Estimation of the L{\'e}vy triplet}
\label{estimation}
In the sequel, we suppose that we are given by the observations \(X_{t_0},X_{t_1},\ldots, X_{t_n}\) of the process \eqref{GOU}  at the equidistant time points \( 0=t_0<t_{1} < \ldots < t_{n}\) with \(t_j=j\cdot \Delta\) for some \(\Delta>0\). Assuming that the process \(X_{t}\) is stationary  (see \cite{Behme}, \cite{Fasen}), we get that the random variables  \(X_{k}:=X_{t_k}, \; k=1,\ldots, n,\) have all the same distribution, which coincides with the distribution of \(A_{\infty}\).

The first step of our estimation procedure consists in the estimation the Laplace exponent \(\phi(z)\) for  \(z=u^{\circ}+\ii v\), where \(\uc>0\) is fixed and \(v\) varies.    An estimator of \(\phi(z)\) can be obtained from  the recursive formula  \eqref{momenty} for the Mellin transform of \(\pi\). In \cite{CPY}, this formula is proved  for real positive \(z\) such that \(\phi(z)>0\) and \(\M(z+1)<\infty\).  The case of complex \(z\) is considered in \cite{Kuznets}, where one can also find  some generalizations of the formula \eqref{momenty} to  the integrals with respect to Brownian motion with drift. In particular, applying Theorem 2 from \cite{Kuznets}, we get that  \eqref{momenty} holds for any \(z \in \Upsilon\). In the situation when  \((\xi_{t})\) is a subordinator, the set \(\Upsilon\) coincides with the positive half-plane (equivalently, the parameter \(\theta\)  is equal to infinity) due to 
\[
	\E\left[
		e^{-x \xi_{1}} 
	\right]
	= - \phi (x) = - \c x - x \int_{\R_{+}} e^{-x u} \nu\left( (u, +\infty) \right) du < 0, \quad \forall \; x>0.
\]

Motivated by \eqref{Mz}, we first estimate the Mellin transform \(\M(z) \)  via  its empirical counterpart 
\begin{eqnarray}
\label{step1}
	\M_{n}(z) : 
	=
	\frac{1}{n}
	\sum_{k=1}^{n} X_{k}^{z-1}
\end{eqnarray}
and then define an estimate of the Laplace exponent \(\phi(z)\)  by
\begin{eqnarray}
\label{step2}
	Y_{n} (z) 
	=
	z \frac{\M_{n}(z)}{\M_{n}(z+1)}.
\end{eqnarray}
If the sequence \(X_1,\ldots,X_n\) has some mixing properties, then we can expect that \(Y_n(z)\to \phi(z)\) in probability.
\subsection{Estimation of $\a$ and $\c$}
Under our assumptions, the Laplace exponent of the L\'evy process \((\xi_t)\) can be represented in the form:
\begin{eqnarray}
\label{estproc}
	\phi(\uc + \ii v)
 =  \lambda+\mu \cdot \left( \uc+\ii v \right) -\F[\bar\nu](-v)
 , \qquad v \in \R,
\end{eqnarray}
where \(\lambda:=\int_{\R_{+}} \nu(dx)\) and  \(\bar\nu(dx):= e^{-\uc x} \nu(dx)\). The general idea of the  procedure described below is to  estimate the Laplace exponent \(\phi(\cdot)\) at the points \(z=\uc+\ii v\) with  \(v\in \mathbb{R}\)   and then use  the relation \eqref{estproc} for the estimation of  parameters.  By the Riemann-Lebesque lemma,  \(\F[\bar\nu] (-v) \to  0\) as \(v \to +\infty\) (see, e.g., \cite{Kawata}) and we conclude from \eqref{estproc} that \(  \phi(\uc+\ii v)\)  is approximately (at least for large \(v\)) a linear function  in \(v\) with the slope \(\c\) and the intercept term \(\a\).  This observation suggests that a properly weighted least-squares approach can be applied to estimate \(\c\) and \(\a\). Let $V_{n}$ be a sequence of positive real numbers and $w(\cdot)$ be a
nonnegative weight function supported on $[0,1].$ Define a scaled weight function 
$w_{n}(v)=V_{n}^{-1}w(v/V_{n})$ and introduce the estimators of the parameters \(\lambda\) and \(\mu\) as the solution of the following optimization problem:
\begin{eqnarray*}
(\lambda_{n},\mu_{n}): & = & \argmin_{(\lambda,\mu)}\int_{0}^{\infty}w_{n}(v)\left|Y_{n}(u^{\circ}+\i v)-\mu\cdot (u^{\circ}+\i v)-\lambda\right|^{2}\, dv\\
 & = & \argmin_{(\lambda,\mu)}\int_{0}^{\infty}w_{n}(v)\left\{ \left|\mathrm{Im}\left[Y_{n}(u^{\circ}+\i v)\right]-\mu v\right|^{2} 
 \right. \\
 && \hspace{4cm} \left.
+ \left|\mathrm{Re}\left[Y_{n}(u^{\circ}+\i v)\right]-\lambda-\mu u^{\circ}\right|^{2}\right\} \, dv
\end{eqnarray*}
with $Y_{n}(z)$ defined in \eqref{step2}. The above optimisation problem admits an explicit solution given by
\begin{eqnarray*}
\mu_{n} & = & \frac{\int_{0}^{\infty}w_{n}(v)\mathrm{Im}\left[Y_{n}(u^{\circ}+\i v)\right]\, dv}{\int_{0}^{\infty}vw_{n}(v)\, dv}=\int_{0}^{\infty}w_{\mu,n}(v)\mathrm{Im}\left[Y_{n}(u^{\circ}+\i v)\right]\, dv\\
\lambda_{n} & = & \frac{\int_{0}^{\infty}w_{n}(v)\mathrm{Re}\left[Y_{n}(u^{\circ}+\i v)\right]\, dv}{\int_{0}^{\infty}w_{n}(v)\, dv}-\mu_{n}u^{\circ}\\\ && \hspace{4cm} =
\int_{0}^{\infty}w_{\lambda,n}(v)\mathrm{Re}\left[Y_{n}(u^{\circ}+\i v)\right]\, dv-\mu_{n}u^{\circ}
\end{eqnarray*}
with $w_{\mu,n}(v):=V_{n}^{-2} w_{\mu}(v/V_{n}) $ and $w_{\lambda,n}(v):=V_{n}^{-1}w_{\lambda} (v/V_{n})$,
where
\[
w_{\mu}(\cdot)= c_{1,w}^{-1} w(\cdot), \qquad
w_{\lambda}(\cdot) = c_{0,w}^{-1} w(\cdot), \qquad 
c_{i,w}=\int_{0}^{1}v^{i}w(v)\, dv,\quad i=0,1.
\]
Taking into account the definition of the weight function \(w_{n}(\cdot)\), we get also some equivalent representations of the estimators \(\mu_{n}\) and \(\lambda_{n}\)
\begin{eqnarray*}
\c_{n} &=&
\argmin_{\c} \int_{\eps}^{1} w(\alpha) 
	\Bigl(
		\Im [Y_{n} (u+\ii \alpha V_{n})] 
		- 
		\c \alpha V_{n}
	\Bigr)^{2} d\alpha \\
\a_{n} &:=&	
	\argmin_{\a}\int_{\eps}^{1} w(\alpha) 
	\Bigl(
		\Re [Y_{n} (u+\ii \alpha V_{n})] 
		- 
		\c_{n} u 
		-
		\a
	\Bigr)^{2} 
	d\alpha.
\end{eqnarray*}
In practice, we need to replace the above integrals by sums. To this end,   let the numbers \(\alpha_{1}, \ldots, \alpha_{M}\) constitute an equidistant grid  on the set \([\eps, 1]\) for some  \(\eps>0.\) We 
estimate the Mellin transform \(\M(z)\)   for all \(z\in \{\uc+\ii \alpha_{m} V_{n}, \; m=1,\ldots, M\}\) and \(z\in \{\uc-1+\i \alpha_{m} V_{n}, \; m=1,\ldots, M\}\)  
and so get the estimates of the Laplace exponent at the discrete points \(z=\uc+\ii \alpha_{m} V_{n}\) (see above). Now we  define an estimate of the parameter \(\c\) via 
\begin{eqnarray}
\label{hatcopt}
	\hat{\c}_{n} &:=& \argmin_{\c} \sum_{m=1}^{M}  w(\alpha_{m}) 
	\Bigl(
		\Im [Y_{n} (\uc+\ii \alpha_{m} V_{n})] 
		- 
		\c \alpha_{m} V_{n}
	\Bigr)^{2}\\
\label{hatc}
	&=& 
	 \frac{
	\sum_{m=1}^{M} 
		  w(\alpha_{m})  \alpha_{m} \: \Im [Y_{n} (\uc+\ii \alpha_{m} V_{n})] 
}
{	
	V_{n} \: \cdot \: \sum_{m=1}^{M} 
		  w(\alpha_{m})  \alpha_{m}^{2} 
}.
\end{eqnarray}

Afterwards, we estimate the parameter \(\a\) by 
\begin{eqnarray}
\label{hataopt}
	\hat{\a}_{n} &:=&
	\argmin_{\a} \sum_{m=1}^{M}   w(\alpha_{m})
	\Bigl(
		\Re [Y_{n} (\uc+\ii \alpha_{m} V_{n})] 
		- 
		\hat{\c}_{n} u - \a
	\Bigr)^{2}\\
\label{hata}
	&=&
	\frac{
		\sum_{m=1}^{M} w (\alpha_{m}) 	\Re [Y_{n}(\uc+\ii \alpha_{m} V_{n})] 
	}
	{
		\sum_{m=1}^{M}
		w (\alpha_{m})
	}
		-\hat{\c}_{n} \uc.
\end{eqnarray}
The whole algorithm is described below.\vsp

\begin{bclogo}[couleur=blue!15, logo=\bccrayon]
{Algorithm 1: Estimation of \(\a\) and \(\c\)}
\begin{algorithm}[H]
\SetAlgoLined

\textbf{Data:}
\(n\) observations \(X_{1},\ldots ,X_{n}\) of the GOU process \((X_t)\) observed at equidistant grid \(j\cdot \Delta,\) \(j=1,\ldots, n.\) 

\vspace{0.3cm}
\textbf{Initiate:} 
Fix  \(V_{n} \to \infty\), \(\eps \in (0,1)\) and \(\uc > -1\). 
\\
Set \(\alpha_{j}=\eps+j\cdot\left( 1- \eps\right)/M,\) \(j=1,\ldots,M.\) 
\\
Fix a  function \(w (\cdot) \geq 0\) supported on \([\eps,1]\).  
\\
Denote \(v_{m,n}:= \alpha_{m} V_{n}\).

\vspace{0.3cm}
\textbf{Algorithm:}
\begin{enumerate}
\item
Estimate the Mellin transform \(\M(z) := \E \left[ A_{\infty}^{z-1} \right]\)  \\
for \(z\in\{\uc+\ii v_{m,n},1+\uc+\ii v_{m,n},\, m=1,\ldots,M\}\) via 
\begin{eqnarray*}
		\M_{n}(z) =
	\frac{1}{n}
	\sum_{k=1}^{n} X_{k}^{z-1}.
\end{eqnarray*}
\item Estimate the Laplace exponent
\(\phi(z) := - \log \E \left[
		e^{-z \xi_{1}}
	\right]\) \\
at the points  \(z\in\{\uc+\ii v_{m,n},\, m=1,\ldots,M\}\)  by
\begin{eqnarray*}
	Y_{n} (z) 
	=
	z \frac{\M_{n}(z)}{\M_{n}(z+1)}.
\end{eqnarray*}

\item[3.] Estimate \(\c\) by 
\begin{eqnarray*}
\label{opt}
	\c_{n} :=  
		 \frac{
	\sum_{m=1}^{M} 
		  w(\alpha_{m})  \alpha_{m} \: \Im [Y_{n} (\uc+\ii v_{m,n})] 
}
{	
	V_{n} \: \cdot \: \sum_{m=1}^{M} 
		  w(\alpha_{m})\,  \alpha_{m}^{2} 
}.
\end{eqnarray*}
\item[4.] Estimate \(\a\) by 
\begin{eqnarray*}
	\a_{n} := 
		\frac{
		\sum_{m=1}^{M} w (\alpha_{m}) 	\Re [Y_{n}(\uc+\ii v_{m,n})] 
	}
	{
		\sum_{m=1}^{M}
		w (\alpha_{m})
	}
		-\c_{n} \uc.
\end{eqnarray*}
\end{enumerate}
\end{algorithm} 
\end{bclogo}

\subsection{Estimation of the L{\'e}vy measure $\nu$}
As a result of Algorithm~1, we obtain the estimates \(\c_{n}\) and \(\a_{n}\) of the parameters \(\c\) and \(\a,\) respectively.  Based on \eqref{estproc}, we first define an estimate for the Fourier transform of \(\bar\nu\) via
\begin{eqnarray}
\label{step5}
	 \hat\F[\bar\nu](-v)  =  - Y_{n}(\uc+\ii v) + \c_{n} \cdot (\uc+\ii v)   +  \a_{n}.
\end{eqnarray}
Next we estimate  the measure \(\nu\) by a regularised Fourier inversion formula
\begin{eqnarray}
\label{step6}
	\nu_n(x)  &=&  \frac{e^{\uc x}}{2 \pi}  \int_{\R} e^{\ii v x } \hat\F[\bar\nu](-v) \K (-v / V_{n})\,  dv,
\end{eqnarray}
where \(\K\) is a regularizing symmetric  kernel supported on \([-1,1]\). Note that with a slight abuse of notation, we use \(\nu\) also for the density of the L{\'e}vy measure, and \(\nu_{n}\) for an estimate of this density. In what follows, we  also use the notation \(\bar{\nu}_{n} = e^{-\uc x} \nu_{n}\). The formal description of the algorithm is given below.


\begin{bclogo}[couleur=blue!15, logo=\bccrayon]
{Algorithm 2: Estimation of \(\nu\)}
\begin{algorithm}[H]
\SetAlgoLined

\textbf{Data:}
\(n\) observations \(X_{1},\ldots ,X_{n}\) of the GOU process \((X_t)\) observed at equidistant grid points \(j\cdot \Delta,\) \(j=1,\ldots, n.\) 

\vspace{0.3cm}
\textbf{Initiate:} 
Fix  \(V_{n} \to \infty\) and \(\uc > -1\). 
\\
Set \(\alpha_{m}=-1+2\cdot j/M,\) \(m=0,\ldots,M.\) 
\\
Fix  a regularizing kernel \(\K\)  supported on \([-1,1]\).\\
Denote \(v_{m,n}:= \alpha_{m} V_{n}\).
\vspace{0.3cm}

\textbf{Algorithm:}
\begin{enumerate}
\item[1-2] The first two steps coincide with ones of Algorithm 1.
\item[3.]  Estimate  \(\F[\bar\nu](-v_{m,n}) \) for \(\bar{\nu} (dx) =e^{-\uc x} \nu (dx)\)  by
\begin{eqnarray*}
	 \hat\F[\bar\nu](-v_{m,n})  =  - Y_{n}(u+\ii v_{m,n}) + \c_{n} \cdot (u+\ii v_{m,n})  +  \a_{n}
\end{eqnarray*}
for \(m=0,\ldots,M.\)
\item[4.]
Estimate \(\nu\) by
\begin{eqnarray*}
	\bar{\nu}_n(x)  &=&  e^{\uc x}  \frac{1}{2\pi\cdot (1+M) }  \sum_{m=0}^{M}{e^{\ii v_{m,n} x }} \hat\F[\bar\nu](-v_{m,n}) \K (\alpha_{m}).
\end{eqnarray*}

\end{enumerate}
\end{algorithm}
\end{bclogo}

\begin{rem}
It is a worth mentioning that the estimation Algorithms 1 and 2 can be applied to a more general situation when 
\begin{eqnarray}
\label{xi}
\xi_{t} =
	\c t  + \tau_{t},
\end{eqnarray}
where the process \(\tau_{t}\) is a difference between two subordinators, i.e., \(\tau_{t} = \tau^{+}_{t} + \tau^{-}_{t}\), and  \(\tau^{+}\) and \(\tau^{-}\) are the processes of finite variation with L{\'e}vy measures \(\nu^{+}\) and \(\nu^{-}\) concentrated on \(\R_{+}\) and \(\R_{-},\) respectively. In fact, in this case, the formula \eqref{estproc} still holds with \[\nu(dx) = \I \{x>0\} \nu^{+} (dx) + \I \{x<0\} \nu^{-}(dx).\] 
Therefore, the consequent estimation of \(\c\), \(\a\) and the Fourier transform of the measure \(e^{-\uc x} \nu(dx)\), as well as the estimation of \(\nu\) are still possible. 
\end{rem}

\section{Convergence}
\label{sec4}
In order to analyse the convergence properties of the estimates \(\mu_n,\) \(\lambda_n\) and \(\nu_n\) we need to further specify the class of L\'evy processes \((\xi_t).\)
\label{further}
\begin{defi}
For $s\in\mathbb{N}\cup \{0\}$ and $R>0,$ let ${\cal G}(s,R)$
denote the set of all L\'evy triplets $(\mu,0,\nu)$, such
that  $\nu$ is supported on $\mathbb{R}_{+}$
and
\begin{eqnarray}
\label{cond2}
\max\left\{\nu(\mathbb{R}_{+}), \int_{\R} |v|^{2 s} \left|
		\F[\bar\nu] (v)
	\right|^{2}\,  dv\right\} \leq  R,
\end{eqnarray}
where \(\bar{\nu}(dx)=e^{-u^{\circ}x}\nu(dx).\) 
\end{defi}
Note that if  \eqref{cond2} holds, then \(\bar\nu\) is $s$-times (weakly) differentiable
with 
\begin{eqnarray}
\label{cond3}
\bigl\Vert \bar{\nu}^{(s)} \bigr\Vert _{\infty}\leq \frac{1}{2\pi}\int_{\R} |v|^s \left| \F[\bar\nu] (-v) \right|\, dv<\infty.
\end{eqnarray}
\par
It turns out that the convergence rates of the estimates \(\mu_n,\) \(\lambda_n\) and \(\nu_n\) crucially depend on the asymptotic behaviour of the Mellin transform of \(A_\infty.\) In order to specify this behaviour, let us  fix some $u^{0}>0$ and introduce two classes of probability densities:
\begin{eqnarray}
\label{P}
\mathcal{P}(\beta,L) & := & \left\{ p:\,\liminf_{|v|\to\infty}\left[|v|^{\beta}\left|\mathcal{M}[p](u^{\circ}+\i v)\right|\right] \geq L\right\} ,\\
\label{E}
\mathcal{E}(\alpha,L) & := & \left\{ p:\,\liminf_{|v|\to\infty}\left[e^{\alpha|v|}\left|\mathcal{M}[p](u^{\circ}+\i v)\right|\right] \geq L\right\},
\end{eqnarray}
where \(\alpha, \beta \in \R,\)  \(L>0\) and for any density \(p,\) \(\mathcal{M}[p]\) stands for the Mellin transform of \(p.\)
Before we formulate the main convergence results, let us look at some examples. 

\begin{ex} Consider the class of L{\'e}vy processes with \(\c=0,\) \(\sigma=0\) and the L{\'e}vy  density \(\nu\) of the form
\[
\nu(x) =  \sum_{j=1}^{N} \left[\sum_{k=1}^{{m_{j}}} g_{jk} x^{k-1}\right] e^{-\rho_{j}x} \cdot \I\{x>0\} 
\]
with \(N, m_{j} \in \N\), \(\rho_{j}>0\), \(g_{jk}>0\). First note that the assumption \eqref{cond1} obviously holds. Let us now check \eqref{cond2}.  We can apply the well-known Erd{\'e}lyi lemma to derive
\[
	\int_{\R_{+}} x^{k-1} f(x)  e^{i v x} dx \asymp  c_{1} v^{-k}, \qquad v \to \infty 
\]
for any exponentially decaying and smooth function \(f\) on \(\R_+\), and some complex \(c_{1}\) depending on \(f\).  Therefore, we conclude that 
\begin{multline*}
	\left| \F[\bar\nu] (-v) \right| 
	=
	\left| 
		\sum_{j=1}^{N} \sum_{k=1}^{{m_{j}}} \alpha_{jk}
		\int_{\R_{+}}  x^{k-1} f_{j} (x)
		e^{i v x}
		dx
	\right| 
	 \asymp c_{2} v^{-k^{*}}, \\ 
	 \mbox{where} \quad
	f_{j}(x)= e^{- (\rho_{j}+u^{\circ}) x},
 \quad k^{*}:=\argmin_{k} 
	\left\{
		\exists \: j: \; \alpha_{jk} \ne 0 
	\right\},
\end{multline*}
where \(c_{2}>0\) depends on \(u^{\circ}\). Hence for any \(s<k^{*}-1\), the condition \eqref{cond2} holds for some \(R>0\).
Furthermore, taking into account the asymptotic behaviour of the Gamma function (see, e.g., formula~8.328 from \cite{Jeffrey}):
\begin{eqnarray}
\label{GR}
 	\left| \Gamma(u+\ii v) \right|= \exp\left\{
		-\frac{\pi}{2} v +\left(u - \frac{1}{2} \right) \ln  v 
	\right\} \cdot \sqrt{2 \pi} \left( 1 + o(1) \right), \qquad v \to \infty, 
\end{eqnarray}
we derive  
\begin{eqnarray*}
\left| 
	\M(\uc + \i v )
\right|
\asymp \sqrt{2 \pi} A^{1- \uc} \exp \left\{ 
	- \frac{\pi}{2} v + 
	\left( 
		\uc - \frac{1}{2} + \sum_{j=1}^{N} \rho_{j}m_{j} + \sum_{j=1}^{K} \Re(\zeta_{j})
	\right) \ln v
\right\},
\end{eqnarray*}
where \(\zeta_{1},\ldots, \zeta_{K}\) are the roots of the equation 
\begin{eqnarray*}
\sum_{j=1}^{N} \sum_{k=1}^{m_{j}} \frac{
	g_{jk} (k-1)!
	}{
		\left( \rho_{j} +z \right)^{k} 
	}= \lambda - \mu z,
\end{eqnarray*}
see \cite{Kuznets2}.
Therefore, for any  \(\uc>1/2\), we conclude that \( \pi \in \mathcal{E}(\pi /2 , L)\) with any \(L>0\).

\end{ex}
\begin{ex}\label{exx} Next, we provide an example of a L{\'e}vy process \(\xi_t\) with \(A_\infty=\int_0^\infty e^{-\xi_t}\,dt\) having a density from  \(\mathcal{P}(\beta,L)\). Consider a subordinator \(\T\) with drift \(\c>0\) and the L{\'e}vy density 
\begin{eqnarray*}
\nu(x)= ab \exp\{-bx\}\: I\{x>0\}, \quad a, b >0.
\end{eqnarray*}
 The exponential functional \(A_\infty\) of the process \((\xi_t)\) has a density of the form
\begin{eqnarray*}
	\pi(x) = C_{1} x^{b} (1- \c x)^{(a/\c) -1 } \: I\{0< x<1/\c\}
\end{eqnarray*}
with some \(C_{1}>0\),  see \cite{CPY}.  In other words,  \(A_{\infty}\) has the same distribution as \(\xi/\c\), where the r.v. \(\xi\) has  the Beta distribution  with parameters \(\alpha=b+1\) and \(\beta = a/ \c = \lambda/ \c\). The Mellin transform of the function \(\pi(x)\) in the half-plane \(\Re(s)>-\alpha\) is hence given by 
\begin{eqnarray*}
	\M(z) = \frac{\E \left[
		\xi^{z-1}
	\right]}{\c^{z-1}}
	 &=&
	 \frac{1}{\c^{z-1}}
	\frac{
		B (z+\alpha - 1, \beta)
	}{
		B (\alpha, \beta)
	}\\
	&=&
	\frac{
		\Gamma (\alpha + \beta)
	}{
		\Gamma (\alpha)
	}
	\cdot 
	\frac{1}{\c^{z-1}}
	\frac{
		\Gamma (z+\alpha - 1)
	}{
		\Gamma (z+\alpha+\beta - 1)
	}.
\end{eqnarray*}
Using \eqref{GR}, we conclude that  the  Mellin transform of \(A_\infty\) has a polynomial decay in this case. More precisely, 
\[ |\M(\uc+ \i v)|  \asymp L\cdot |v|^{-\lambda/\c}  \quad \mbox{ with } \quad L=\c^{-\uc+1}  \frac{\Gamma(\lambda/\mu+b+1)}{\Gamma(b+1)},\] 
 as \(|v|\to\infty\) and therefore \(\pi \in \mathcal{P}(\lambda/ \mu, L)\).
\end{ex}

Let us now formulate the main result concerning the convergence  of the estimates  \(\mu_{n}\) and \(\lambda_{n}\).

\begin{thm}[upper bounds for $\c_{n}$ and $\a_{n}$]
\label{corub} 
Let \((\xi_{t})\) be a L{\'e}vy process with a triplet from $\mathcal{G} (s,  R)$. Suppose that the sequence  $X_{0},X_{1},\ldots,X_{n}$   is $\alpha$-mixing  and strictly stationary.  Denote the \(\alpha\)-mixing coefficients of the sequence $X_{0},X_{1},\ldots,X_{n}$  by \(\alpha(s)\). 
\begin{enumerate}[(i)]
\item  Assume that the  density \(\pi\) of \(A_\infty\) belongs to \(\mathcal{P}(\beta,L)\) with some \(\beta \in \R\) and \(L>0\), and moreover 
\begin{eqnarray}
\label{expcond}
\alpha(j)\lesssim e^{-j\alpha^{*}}, \quad j\in\mathbb{N}, \quad\mbox{for some} \quad \alpha^{*} \geq 0.
\end{eqnarray}
Then the quadratic risks of the estimates \(\mu_n\) and \(\lambda_n,\) under the choice \(V_{n}=n^{1/(2\beta+2s+3)},\) satisfy the following asymptotic relations
\[
\mathrm{E}\left[\left|\mu_{n}-\mu\right|^{2}\right]\lesssim n^{-2(s+2)/(2\beta+2s+3)}\log(n)
\]
and 
\[
\mathrm{E}\left[\left|\lambda_{n}-\lambda\right|^{2}\right]\lesssim n^{-2(s+1)/(2\beta+2s+3)}\log(n),
\]
as \(n\to\infty.\)
\item 
If  \(\pi\in \mE(\alpha, L)\) and
\begin{eqnarray}
\label{expcond_pol}
\alpha(j)\lesssim j^{-\alpha^{*}},\quad j\in\mathbb{N}, \quad \mbox{for some} \quad \alpha^{*} \geq 2,
\end{eqnarray}
then the choice 
\[V_{n}=\frac{1}{2\alpha}\log(n)-\frac{s+2}{\alpha}\log(\log(n)),\] leads to the rates 
\[
\mathrm{E}\left[\left|\mu_{n}-\mu\right|^{2}\right]\lesssim\log^{-2(s+2)}(n),
\]
\[
\mathrm{E}\left[\left|\lambda_{n}-\lambda\right|^{2}\right]\lesssim\log^{-2(s+1)}(n).
\]
\end{enumerate}
\end{thm}
\begin{proof}
	Proof is given in Section~\ref{upplmu}. 
\end{proof}
In a similar way, we can establish the upper bounds for the risk of  \(\bar\nu_{n}\). In the theorem formulated  below, the quality of the estimate \(\bar\nu_{n}\) is measured in terms of the mean integrated squared error (MISE)
\begin{eqnarray*}
\textrm{MISE} (\bar{\nu}_{n}) &:=& 
\E \left[ 
	\int_{\R} \left|
		\bar{\nu}_{n}(x) - \bar\nu(x)
	\right|^{2}  dx
\right].
\end{eqnarray*}

\begin{thm}[upper bounds for $\bar\nu_{n}$]
\label{nuub}
 Let the assumptions of Theorem~\ref{corub} be fulfilled and let   \(\K(\cdot)\) be a kernel satisfying
\begin{eqnarray}
\label{assk} 
	|1 - \K(x)| \leq A |x|^{s}, \qquad \forall x \in \R \setminus \{0\}
\end{eqnarray}
with some \(A>0.\)
\begin{enumerate}[(i)]
\item  Assume that the  density of \(A_\infty\) belongs to \(\mathcal{P}(\beta,L)\) with some \(\beta \in \R\) and \(L>0\), and moreover 
\begin{eqnarray*}
\alpha(j)\lesssim e^{-j\alpha^{*}},\quad j\in\mathbb{N}, \quad \mbox{for some} \quad \alpha^{*} > 0.
\end{eqnarray*}
Then under the choice \(V_{n}=n^{1/(2\beta+2s+3)}\),  the MISE of the estimator \(\bar\nu_n\) is bounded as follows:
\[
\textrm{MISE} (\bar{\nu}_{n}) \lesssim n^{-2s/(2\beta+2s+3)},\quad n\to\infty.
\]
\item 
If the density of \(A_\infty\) belongs to the class \(\mE(\alpha, L)\) and 
\begin{eqnarray*}
\label{expcond_pol}
\alpha(j)\lesssim j^{-\alpha^{*}},\quad j\in\mathbb{N}, \quad \mbox{for some} \quad \alpha^{*} \geq 2,
\end{eqnarray*}
then under the choice
\[V_{n}=\frac{1}{2\alpha}\log(n)-\frac{s+2}{\alpha}\log(\log(n))\]
we have 
\[
\textrm{MISE} (\bar{\nu}_{n}) \lesssim\log^{-2 s}(n),\quad n\to\infty.
\]
\end{enumerate}
\end{thm} 
\begin{proof}
	Proof is given in Section~\ref{uppmise}. 
\end{proof}

The next theorem shows that the rates obtained in the previous theorem are optimal up to a logarithmic factor. 

\begin{thm}[lower bounds for $\bar\nu_{n}$] 
\label{thm3} 
Fix some  \(s \in \N\cup\{0\},\) \(R>0,\) \(\alpha>0,\) \(\beta>0,\)  \(L>0\) and define 
\begin{eqnarray*}
\varphi_{n}(\pi):=\varphi_{n}(\pi,\rho)=
\begin{cases}
		n^{s/(2\beta+2s+3)}\log^{-\rho}(n) , &\text{if $\pi \in \mathcal{P}(\beta,L)$,}\\
		 \log^{ s}(n), &\text{if $\pi \in \mathcal{E}(\alpha,L)$,}
\end{cases}
\end{eqnarray*}
for any  \(\rho>0\) and any probability density \(\pi \in \mathcal{P}(\beta,L)\cup \mathcal{E}(\alpha,L), \) 
Then  for some \(\rho^{*}>0\), it holds 
\begin{eqnarray}
\label{lowerbound}
\inf_{\bar{\nu}_{n}} \sup_{\substack{\mathcal{T}\in\mg(s,R)\\ \pi_{\mathcal{T}} \in \mathcal{P}(\beta,L)\cup \mathcal{E}(\alpha,L)} }  
\left\{ 
\varphi_{n}^{2}(\pi_{\mathcal{T}}, \rho^{*})\cdot 
\E_{\pi^{\otimes n}_{\mathcal{T}}} \left[ 
	\int_{\R} \left|
		\bar{\nu}_{n}(x) - \bar\nu(x)
	\right|^{2}  dx
\right]
\right\} > 0,
\end{eqnarray}
where   the infimum is taken over all possible estimates \(\bar{\nu}_{n}\) of the function \(\bar\nu\) based on i.i.d. sample \(X_1,\ldots,X_n\) from the distribution \(\pi_{\mathcal{T}}\) of \(A_{\infty}:= \int_{0}^{\infty} e^{-\xi_{t}} \; dt\) such that the L\'evy triplet \(\mathcal{T}\) of \((\xi_t)\) belongs to \(\mg(s,R).\) 

\end{thm}
\begin{proof}
	Proof is given in Section~\ref{prooflb}. 
\end{proof}

An important condition of Theorems ~\ref{corub} and \ref{nuub}  is \eqref{expcond}, which means that the sequence \(X_{0}, X_{1},\ldots, X_{n}\) is exponentially \(\alpha\)-mixing. Since \(\beta\)-mixing coefficient  between two sigma-algebras is larger than or equal to the corresponding \(\alpha\)-mixing coefficient, it is sufficient to show that \(X_{0}, X_{1},\ldots, X_{n}\) is an exponentially \(\beta\)-mixing sequence (see Section~1.1 from \cite{Bosq}). For the case of the GOU processes \eqref{GOU}, the latter question  was addressed in \cite{Fasen}. The sufficient conditions for exponential \(\beta\)-mixing given in \cite{Fasen} are: 
\begin{enumerate}
\item the distribution of \(A_{\infty}\) has a Pareto-like asymptotic behaviour, that is,
\begin{eqnarray*}
\P \left\{ 
	A_{\infty}>x 
\right\}
\asymp C x^{-\alpha}
\quad \mbox{as}\quad x \to \infty
\end{eqnarray*}
with some \(\alpha>0\) and  \(C>0;\)
\item there exist \(A>0\), \(B>A\) and \(h>0\) such that \(\psi(A)=0, \psi(B)<\infty\) with \(\psi\) given in \eqref{phis}, and 
\begin{eqnarray*}
\E \left| 
e^{-\xi_{h}}\int_{0}^{h}e^{\xi_{u-}} du
\right|^{B} <\infty.
\end{eqnarray*}
\end{enumerate}
As it is proved in \cite{LindnerMaller}, both conditions are guaranteed by the positiveness of \(\mu\) and the existence of a positive zero of the function \(\psi(\cdot)\). We refer also to  \cite{Lee} for some further results in this direction.

\section{Simulation study} 
\label{secsim}

\sec{Example 1.} 
Consider the subordinator \(\tau_{t}\) with the L{\'e}vy density 
\begin{eqnarray}
\label{nu}
\nu(x)= ab \exp\{-bx\}\: \I\{x>0\}, \quad a, b >0.
\end{eqnarray}
Note that in this case, \(\lambda=\int_{\R_{+}} \nu (u) du = a\).  
Define a L\'evy process 
\begin{eqnarray}
\label{xiex}
\xi_{t}=\mu t + \sigma W_{t} + \tau_{t},
\end{eqnarray}
where \(W_{t}\) is a Brownian motion.  The Laplace exponent of \(\xi_{t}\) is given by
\begin{eqnarray}
\label{phiss2}
\phi(z) = z \left(
		\c - \frac{1}{2} \sigma^{2} z + \frac{a}{b+z}
	\right).
\end{eqnarray}

 In \cite{CPY}, it is shown that the exponential functional \(A_{\infty}=\int_{0}^{\infty} e^{-\xi_{t}} \; dt\) is finite for any \(\mu\) and \(\sigma\), and moreover  the density function \(\pi\) of \(A_{\infty}\) satisfies the following differential equation 
\begin{multline}
 -\frac{\sigma^{2}}{2} x^{2} \pi''(x) + 
 \left[
 	\left(
		\frac{\sigma^{2}}{2} (3-b) +\c
	\right)
	x
	-1
\right] \pi'(x) \\
+
\left[
	\left( 
		1-b
	\right)
	\left(
	 	\frac{\sigma^{2}}{2} + \c
	\right)
	-a+\frac{b}{x}
 \right] \pi(x) = 0.
 \label{eq}
\end{multline}
Some special cases are considered below:
\begin{enumerate}
\item  In the case \(\c=0, \: \sigma=0\) (pure jump  process), this equation has a solution
\begin{eqnarray}
\label{kx1}
	\pi_{1}(x) = C x^{b} e^{-a x} \: I\{x>0\},
\end{eqnarray}
and therefore \(A_{\infty} \eqd  G(b+1, a)\), where \(G(\alpha, \beta)\) is  a Gamma distribution with shape parameter \(\alpha\)  and rate \(\beta\). 
\item If \(\c>0, \: \sigma=0\)  (pure jump process with drift), then 
\begin{eqnarray}
\label{kx2}
	\pi_{2}(x) = C x^{b} (1- \c x)^{(a/\c) -1 } \: I\{0< x<1/\c\}.
\end{eqnarray}
In this situation \(A_{\infty} \eqd  B(b+1, a/\c) / \c\), where \(B(\alpha, \beta)\) is a Beta - distribution.
\item In the case \(\c \ne 0, \: \sigma \ne 0\),  the equation \eqref{eq} also allows for the closed form solution. Assuming for simplicity \(\sigma^{2}/2 =1\), \(\c=-(b+1)\), we get the solution of \eqref{eq} in the following form:
\begin{eqnarray}
\label{kx}
	\pi_{3}(x)=C\: x^{b-1/2} \exp\left\{\frac{1}{2x}\right\} I_{\mu}\left(\frac{1}{2x}\right),
\end{eqnarray}
where we denote by \(I_{\mu}\) the modified Bessel function of the first kind, \(\mu=\sqrt{a+1/4}\), and the constant \(C\) is later chosen to guarantee the condition \(\int_{0}^{\infty}\pi_{3}(x) dx =1\). \vsp
\end{enumerate}
For our numerical study, we assume that the data are generated from  the distribution of \eqref{Ainfty}, where the process \((\xi_{t})\) is defined by \eqref{xiex} with 
\(\c=1.8, \sigma=0,\)  and the subordinator \(\tau_{t}\) in the form \eqref{nu} with \(a=0.7\), \(b=0.2\). A sample from the distribution of the integral \(A_{\infty}\) can be simulated from the corresponding Beta-distribution, see \eqref{kx2}. 
In the first step, we estimate the Mellin transform \(\M(z)\)  for \(z=u+\ii v\) with \(u=\uc=29\) and \(u=\uc+1=30\) and  \(v\) lying on the equidistant grid between \(-30\) and \(30\). Next, we estimate the Laplace exponent of \(\xi\) by the formula \eqref{step2}. Figure~\ref{plot1} graphically compares the proposed estimator of the Laplace exponent \(\phi(\uc+\i v)\) with its theoretical values \((\mu+a/(b+\uc+\i v))\cdot \left(\uc+\i v\right)\) . 
\begin{figure}
\begin{center}
\includegraphics[width=0.6\linewidth ]{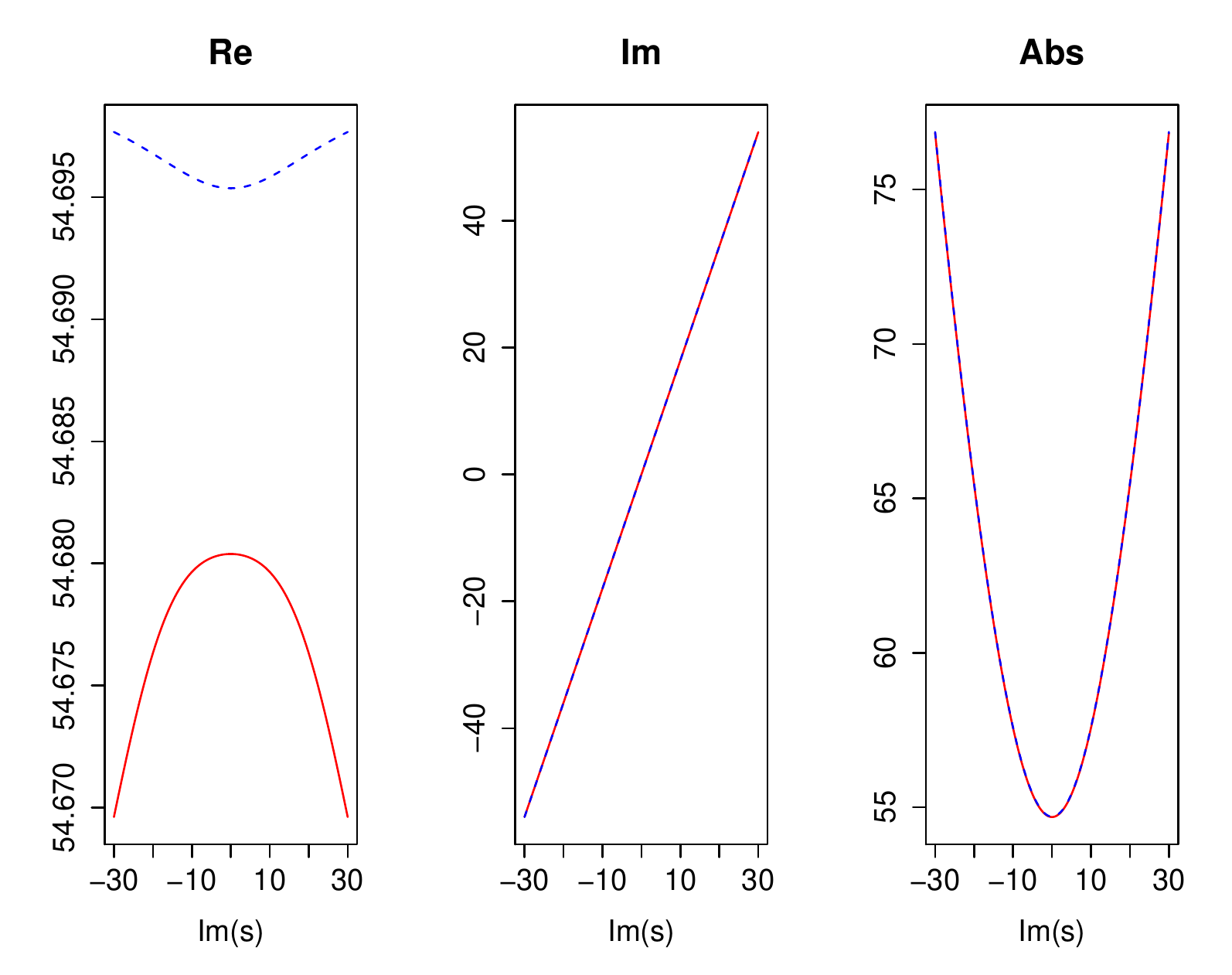}\caption{Plots of theoretical (blue dashed) and empirical (red solid) Laplace exponents in Example 1. Real, imaginary parts and absolute values are presented.   \label{plot1}}
\end{center}
\end{figure}

\begin{figure}
\begin{center}
\includegraphics[width=0.6\linewidth ]{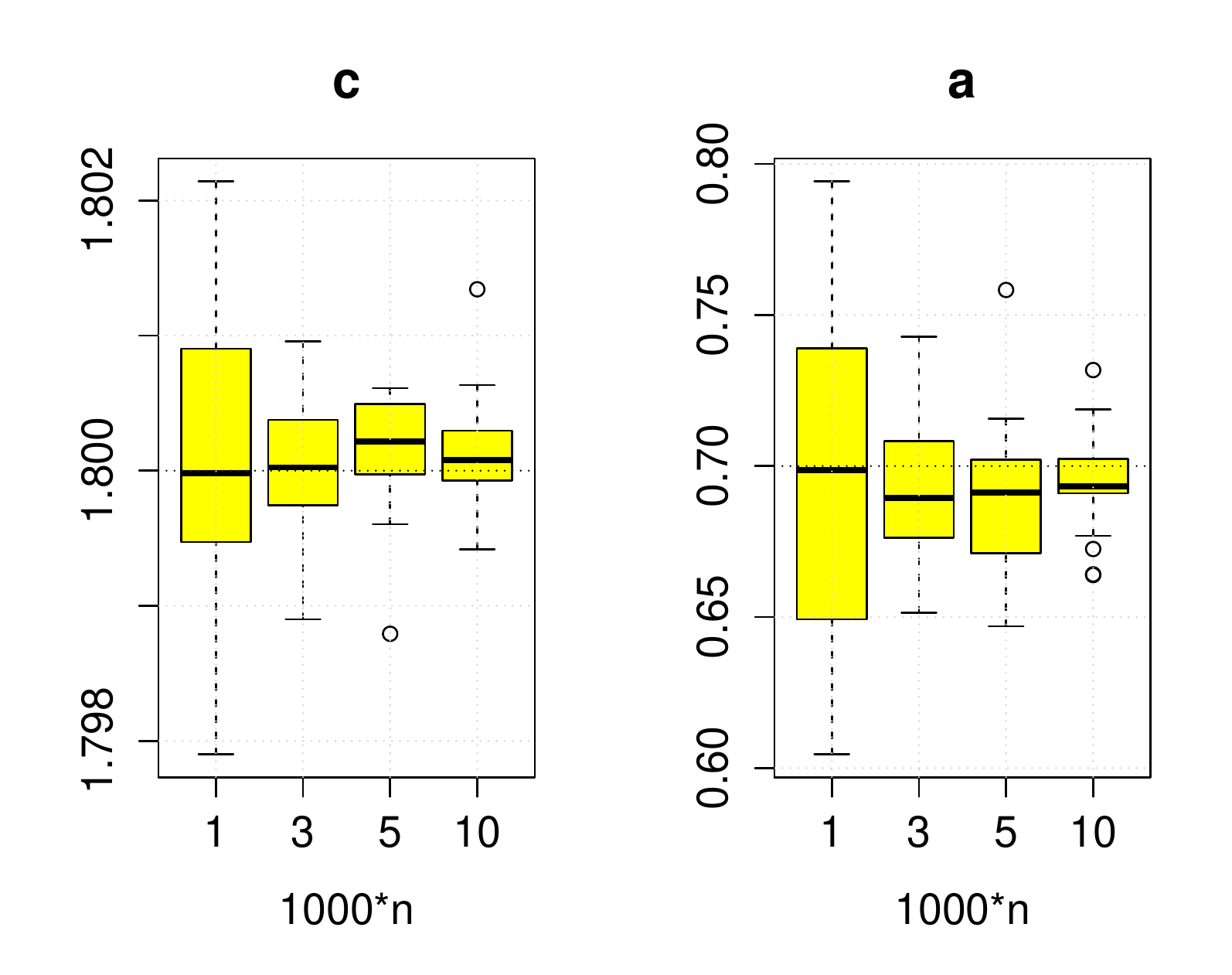}\caption{Boxplots for the estimates of \(\c=c\) and \(\a=a\) for different sample sizes \(n\) based on 25 simulation runs.\label{plot2}}
\end{center}
\end{figure}
Estimates for the parameters \(\c\) and \(\a=a\) are given in   \eqref{hatc} and \eqref{hata}, respectively. The boxplots of this estimates based on 25 simulation runs are presented on Figure \ref{plot2}. \vsp

\sec{Example 2.}  Consider the compound Poisson process 
\[
	\xi_{t} = -\log q \left(\sum_{k=1}^{N_{t}} \eta_{k} \right),
\]
where \(q \in (0,1)\) is fixed, \(N_{t}\) is a Poisson process with intensity \(\lambda\) and \(\eta_{k}\) are i.i.d. random variables with a distribution \(\L\). The integral \(A_{\infty}\) admits the representation 
\begin{eqnarray*}
	A_{\infty} = \int_{0}^{\infty} q^{-\xi_{t}} dt = \sum_{n=0}^{\infty} q^{S_{n}} \left(T_{n+1} - T_{n}\right),
\end{eqnarray*}
where \(T_{n}\) is the jump time of \(N,\) i.e., \(T_{n} = \inf\left\{ t: N_{t} =n \right\}\), and \(S_{n}=\sum_{k=1}^{n}\eta_{k}.\) Note that if \(\eta_{k}\) take only positive values,  then \(\xi_{t}\) is a subordinator. For the overview of the properties of the integral \(A_{\infty}\) in the particular case \(\eta_{k} \equiv 1\) (that is, \(\xi_{t}\) is a Poisson process up to a constant), we refer to \cite{BertoinYor}.

Fix some positive \(\alpha\) and consider the case when \(\L\) is  the standard normal distribution truncated on the interval \((\alpha, +\infty)\).  The density function of \(\L\) is given by \[p_{\L} (x)= p(x) / (1-F(\alpha)),\] where \(p(\cdot)\) and \(F(\cdot)\) are the density and the distribution functions of the standard Normal distribution. In this case, the Laplace exponent  of \(\xi_{t}\) is equal to 
\begin{eqnarray*}
	\phi(z) = \lambda \left[
		1 -  \frac{	
			1-F\left(\alpha + (\log q) z\right)
		}
		{
			1 - F\left(\alpha\right)
		}\;
		\exp \left\{
		-
			\frac{
				\left(
					\log q 
				\right)^{2}
				z^{2}
			}{2}
		\right\}
	\right],
\end{eqnarray*}
where the function \(F(\cdot)\) can be calculated for complex arguments from the error function: 
\[
F(z) := \frac{1}{2} \left( \erf\left( \frac{z}{\sqrt{2}} \right) + 1 \right), \quad \mbox{where} \quad 
\erf(z) = \frac{2}{\sqrt{\pi}}\int_{0}^{z} e^{-s^{2}} ds.
\]
In this example, we aim to  estimate the L{\'e}vy measure of the process \(\-(\xi_{t})\), which is given by
\[
\nu(dx)  = \frac{\lambda}{1-F(\alpha)} \: p(x)  \I\{x>\alpha\} dx.
\]
For our numerical study, we take \(q=0.5, \alpha=0.1\), and \(\lambda=1\). First, we estimate the Laplace exponent by  \eqref{step2}. The quality of the corresponding estimate  at the complex points \(z=\uc+\ii v\) with \(\uc=1\) and \(v \in [-5,5]\) can be visually seen in Figure~\ref{fig3}. 
\begin{figure}
\begin{center}
\includegraphics[width=0.6\linewidth ]{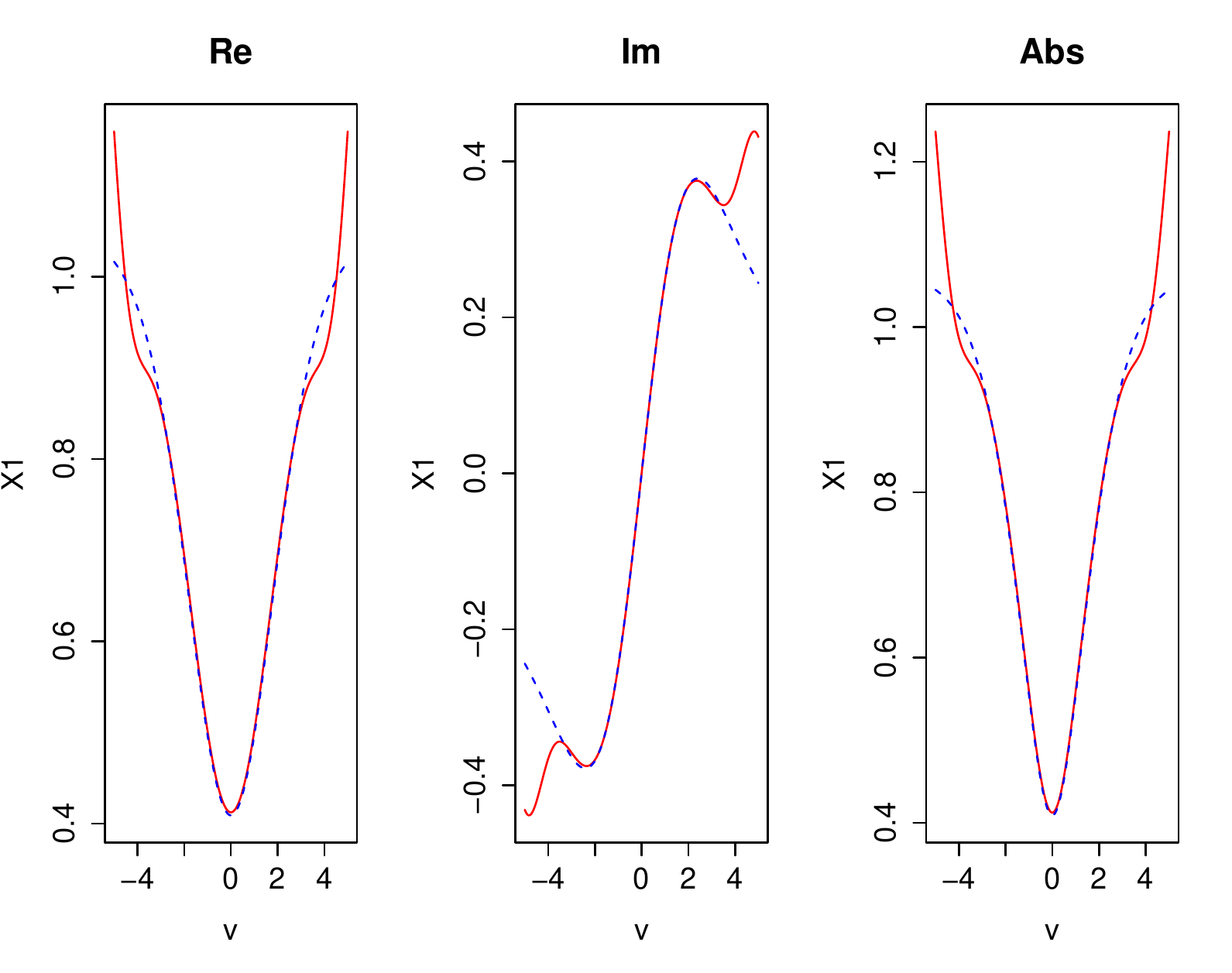}\caption{Plots of theoretical (blue dashed) and empirical (red solid) Laplace exponents for Example 2. Graphs present real, imaginary and absolute values. For  \(v \in[-3,3]\) the curves are visually indistinguishable. \label{fig3}}
\end{center}
\end{figure}
\begin{figure}
\begin{center}
\includegraphics[width=0.6\linewidth ]{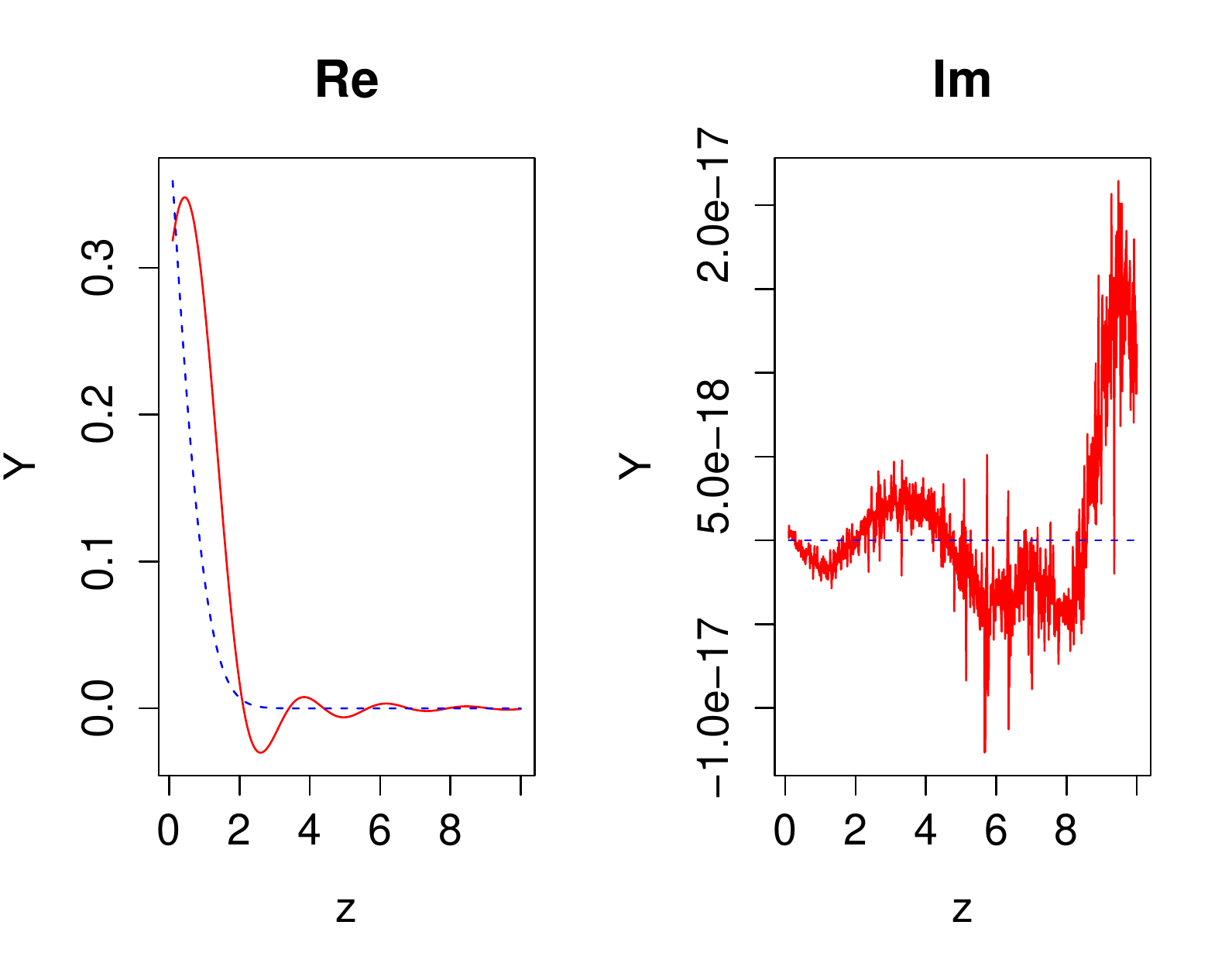}\caption{Left: plots of the L{\'e}vy density  (blue dashed line) and its estimate \(\bar{\nu}_n(z)\) (red solid). Right: the imaginary part of the estimate \(\bar{\nu}_n(z)\) (red solid) and the line \(Y=0\) (blue dashed line). \label{fig4}}
\end{center}
\end{figure}
Next, we proceed with the estimation of the Fourier transform of the measure \(\bar{\nu}(x):=e^{-\uc x} \nu(x)\)
 by applying \eqref{step5}. For the last step of the Algorithm~2, i.e. the reconstruction of the L{\'e}vy measure by \eqref{step6},  we follow  \cite{Belomest2011}  and use the so-called flat-top kernel, which is defined as follows:	
\begin{eqnarray*}
\K(x)=
\begin{cases}
1, & |x|\leq 0.05, \\
\exp\left( -\frac{e^{-1/(|x|-0.05)}}{1-|x|} \right), & 0.05<|x|<1,\\
0, &  |x|\geq 1.
\end{cases}
\end{eqnarray*}
The quality of the resulting  estimate \(\bar{\nu}_n\)  is shown in Figure~\ref{fig4}.

\section{Proofs}\label{theory}
\subsection{Upper bounds for the quadratic risks of  $\mu_{n}$ and $\lambda_{n}$} \label{upp} 
\label{upplmu}

The next proposition is the main technical result for this section.
\begin{prop}
\label{propub} 
Let \(\xi_{t}\) be a L{\'e}vy triplet from ${\cal G} (s,R)$. Suppose that the sequence  $X_{0},X_{1},\ldots,X_{n}$ of observations of the exponential functional \(A_{\infty}:=\lim_{T\to\infty} A_{T}= \int_{0}^{\infty} e^{-\xi_{t}} \; dt\)  is $\alpha$-mixing  and strictly stationary.  Denote the mixing coefficients of the sequence $X_{0},X_{1},\ldots,X_{n}$  by \(\alpha(s)\).

Then for
any $p\in\{0,1,\ldots,n\}$ we have 
\begin{multline*}
\mathrm{E}\left[\left|\mu_{n}-\mu\right|^{2}\right]  \lesssim 
\frac{p}{n}\int_{0}^{\infty}\frac{|u^{\circ}+\i v|^{2}}{\left|\M(u^{\circ}+\i v+1)\right|^{2}}\left|w_{\mu,n}(v)\right|^{2}\, dv\\
+\sum_{j=p+1}^{n}\alpha(j)\left[\int_{0}^{\infty}\frac{\left|u^{\circ}+\i v\right|\left|w_{\mu,n}(v)\right|}{\left|\M_{n}(u^{\circ}+\i v+1)\right|}\, dv\right]^{2}
\\
  +  \|\bar{\nu}^{(s)}\|_{\infty}^{2}\|\mathcal{F}^{-1}[w_{\mu,n}(\cdot)/(-i\cdot)^{s}]\|_{L^{1}}^{2},
  \end{multline*}
  \begin{multline*}
\mathrm{E}\left[\left|\lambda_{n}-\lambda\right|^{2}\right]  \lesssim  \frac{p}{n}\int_{0}^{\infty}\frac{|u^{\circ}+\i v|^{2}}{\left|\M(u^{\circ}+\i v+1)\right|^{2}}\left|w_{\lambda,n}(v)\right|^{2}\, dv\\+\sum_{j=p+1}^{n}\alpha(j)\left[\int_{0}^{\infty}\frac{\left|u^{\circ}+\i v\right|\left|w_{\lambda,n}(v)\right|}{\left|\M_{n}(u^{\circ}+\i v+1)\right|}\, dv\right]^{2}\\
+ \|\bar{\nu}^{(s)}\|_{\infty}^{2}\|\mathcal{F}^{-1}[w_{\lambda,n}(\cdot)/(-i\cdot)^{s}]\|_{L^{1}}^{2},
\end{multline*}
provided $\sum_{j=1}^{\infty}\alpha^{1-\epsilon}(j)<\infty$ for some
$\epsilon>0$ and the sequence $V_{n}$ satisfies 
\begin{equation}
\sup_{v\in[0,V_{n}]}\frac{1}{\left|\M(u^{\circ}+\i v+1)\right|}=o(n^{1/2}).\label{eq:Vn_cond}
\end{equation}
\end{prop}
\begin{proof}
\textbf{1.} Denote $Y(z) :=\phi(z) = z \cdot \M(z) / \M(z+1),$ then 
\[\mu=\int_{0}^{\infty}w_{\mu,n}(v)\Im\left[Y(u^{\circ}+\i v)\right]\, dv + \int_{0}^{\infty}w_{\mu,n}(v)\mathrm{Im}[\mathcal{F}[\bar{\nu}](-v)]\, dv\]
and we have 
\begin{eqnarray*}
\mu_{n}-\mu & = & \int_{0}^{\infty}w_{\mu,n}(v)\mathrm{Im}\left[Y_{n}(u^{\circ}+\i v)-Y(u^{\circ}+\i v)\right]\, dv\\
&& \hspace{4cm} - \int_{0}^{\infty}w_{\mu,n}(v)\mathrm{Im}[\mathcal{F}[\bar{\nu}](-v)]\, dv\\
 & = & \mathrm{Im}\left[\int_{0}^{\infty}w_{\mu,n}(v)S_{n}(u^{\circ}+\i v)\, dv\right] - \mathrm{Im}[D_{n}(u^{\circ})]
\end{eqnarray*}
with 
\[
S_{n}(u^{\circ}+\i v)=Y_{n}(u^{\circ}+\i v)-Y(u^{\circ}+\i v),\quad  D_{n}(u^{\circ})=
\int_{0}^{\infty}w_{\mu,n}(v) \mathcal{F}[\bar{\nu}](-v) dv.
\]
Note that 
\begin{eqnarray*}
\E \left[ \left(
\mu_{n}-\mu
\right)^{2}
\right] \leq 2 \cdot
\E\left[
	\left(
		\mathrm{Im}\left[\int_{0}^{\infty}w_{\mu,n}(v)S_{n}(u^{\circ}+\i v)\, dv\right]
	\right)^{2}
\right]
+
2
\left|
	D_{n}(u^{\circ})
\right|^{2}.
\end{eqnarray*}

\textbf{2.} Since 
\begin{eqnarray*}
\frac{S_{n}(z)}{z} & = & \frac{\M_{n}(z)}{\M_{n}(z+1)}-\frac{\M(z)}{\M(z+1)}\\
 & = & \frac{\M_{n}(z)\M(z+1)-\M(z)\M_{n}(z+1)}{\M_{n}(z+1)\M(z+1)}\\
 & = & \frac{\left[\M_{n}(z)-\M(z)\right]\M_{n}(z+1)-\left[\M_{n}(z+1)-\M(z+1)\right]\M_{n}(z)}{\M_{n}(z+1)\M(z+1)}\\
 & = & \frac{\left[\M_{n}(z)-\M(z)\right]}{\M(z+1)}-\frac{Y_{n}(z)}{z}\,\frac{\left[\M_{n}(z+1)-\M(z+1)\right]}{\M(z+1)}\\
 & = & \frac{\left[\M_{n}(z)-\M(z)\right]}{\M(z+1)}-\frac{S_{n}(z)}{z}\,\frac{\left[\M_{n}(z+1)-\M(z+1)\right]}{\M(z+1)}\\
 &  & -\frac{Y(z)}{z}\,\frac{\left[\M_{n}(z+1)-\M(z+1)\right]}{\M(z+1)},
\end{eqnarray*}
we get
\[ 
S_{n} \cdot (1+R_{2,n} )=-Y\cdot R_{2,n} +R_{1,n}
\]
with 
\[
R_{1,n}(z)=z\,\frac{\left[\M_{n}(z)-\M(z)\right]}{\M(z+1)},\quad R_{2,n}(z)=\frac{\left[\M_{n}(z+1)-\M(z+1)\right]}{\M(z+1)}.
\]
Following the lines of the proof of Theorem~1.5 from \cite{Bosq}, we get 
\begin{multline}
\label{bosq}
\mathrm{E}\left[\left|\M_{n}(z)-\M(z)\right|^{2}\right] = \frac{1}{n^{2}}
\sum_{0\leq k,j\leq n-1}\mathrm{Cov}\left(X_{k}^{z-1},X_{j}^{z-1}\right)\\
=
\frac{1}{n} \Var\left(X_{0}^{z-1} \right) + 
\frac{2}{n}\sum_{k=1}^{n-1} \left( 1 - \frac{k}{n} \right) 
\mathrm{Cov}\left(X_{0}^{z-1},X_{k}^{z-1}\right).
\end{multline}
Note that the sum in the last representation converges as \(n \to \infty\), because by Davydov's inequality
\begin{eqnarray}
\label{davydov}
	\left| 
		\mathrm{Cov}\left(X_{0}^{z-1},X_{k}^{z-1}\right) 
	\right|
	\leq 
	\frac{2r}{r-2} 
	\left( 2 \alpha(k) \right)^{(r-2)/r}	
	\left(
		\E \left[ 
			X_{0}^{\left( u^{\circ}-1 \right) r}
		\right]
	\right)^{2/ r},
\end{eqnarray}
and therefore the series \(\sum \left(X_{0}^{z-1},X_{k}^{z-1}\right)\) is convergent if \(r=2 / \eps\).

We have $\mathrm{E}\left[\left|\M_{n}(u^{\circ}+\i v)-\M(u^{\circ}+\i v)\right|^{2}\right]\lesssim n^{-1}$
uniformly in $v\in\mathbb{R}.$ As a result
\[
\mathrm{E}\left[\left|R_{2,n}(u^{\circ}+\i v)\right|^{2}\right]\lesssim\frac{1}{n \cdot |\M(u^{\circ}+\i v+1)|^{2}}.
\]
The condition (\ref{eq:Vn_cond}) implies now that $\sup_{v\in[0,V_{n}]}\left|R_{2,n}(u^{\circ}+\i v)\right|^{2}=o_{P}(1).$
Furthermore, we have 
\begin{multline*}
\mathrm{Var}\left[\int_{0}^{\infty}R_{1,n}(u^{\circ}+\i v)w_{\mu,n}(v)\, dv\right] \\= \int_{0}^{\infty}\int_{0}^{\infty}\frac{\mathrm{Cov}(\M_{n}(u^{\circ}+\i v_{1}),\M_{n}(u^{\circ}+\i v_{2}))}{\M(u^{\circ}+\i v_{1}+1)\overline{\M(u^{\circ}+\i v_{2}+1)}}\\ \cdot (u^{\circ}+\i v_{1})(u^{\circ}-\i v_{2})w_{\mu,n}(v_{1})\, w_{\mu,n}(v_{2})\, dv_{1} dv_{2}.
\end{multline*}
Similar to \eqref{bosq},  we consider a representation
\begin{multline*}
\mathrm{Cov}(\M_{n}(u^{\circ}+\i v_{1}),\M_{n}(u^{\circ}+\i v_{2})) \\ 
=\frac{1}{n}\left[g_{0}(v_{1},v_{2})  +
 2\sum_{j=1}^{p}g_{j}(v_{1},v_{2})+2\sum_{j=p+1}^{n-1}g_{j}(v_{1},v_{2})\right],
\end{multline*}
where \(g_{j}(v_{1}, v_{2}) := (1 - j/n) \cdot \mathrm{Cov}\left(X_{0}^{u^{\circ}+\i v_{1}-1},X_{j}^{u^{\circ}+\i v_{2}-1}\right), \; j=0..(n-1)\). Applying once more Davydov's inequality, we get 
\begin{eqnarray}
\label{davydov}
	\left| 
	g_{j}(v_{1}, v_{2}) 
	\right|
	\leq 
	\frac{2r}{r-2} 
	\left( 2 \alpha(j) \right)^{(r-2)/r}	
	\left(
		\E \left[ 
			X_{0}^{\left( u^{\circ}-1 \right) r}
		\right]
	\right)^{2/ r},
\end{eqnarray}
 Now using the Cauchy-Schwarz inequality we get
\begin{eqnarray*}
\mathrm{Var}\left[\int_{0}^{\infty}R_{1,n}(u^{\circ}+\i v)w_{\mu,n}(v)\, dv\right] & \lesssim & p\int_{0}^{\infty}\frac{\left|u^{\circ}+\i v\right|^{2}\left|w_{\mu,n}(v)\right|^{2}}{\left|\M(u^{\circ}+\i v+1)\right|^{2}}\, dv\\
 &  &+\sum_{j=p+1}^{n}\alpha(j)\left[\int_{0}^{\infty}\frac{\left|u^{\circ}+\i v\right|\left|w_{\mu,n}(v)\right|}{\left|\M(u^{\circ}+\i v+1)\right|}\, dv\right]^{2}.
\end{eqnarray*}
 Finally using the fact
$\sup_{v\in[0,V_{n}]}\left|R_{2,n}(u^{\circ}+\i v)\right|^{2}=o_{P}(1),$
we derive
\begin{eqnarray*}
\mathrm{Var}\left[\int_{0}^{\infty}S_{n}(u^{\circ}+\i v)w_{\mu,n}(v)\, dv\right] & \lesssim & p\int_{0}^{\infty}\frac{\left|u^{\circ}+\i v\right|^{2}\left|w_{\mu,n}(v)\right|^{2}}{\left|\M(u^{\circ}+\i v+1)\right|^{2}}\, dv\\
 &  & +\sum_{j=p+1}^{n}\alpha(j)\left[\int_{0}^{\infty}\frac{\left|u^{\circ}+\i v\right|\left|w_{\mu,n}(v)\right|}{\left|\M(u^{\circ}+\i v+1)\right|}\, dv\right]^{2}.
\end{eqnarray*}

\textbf{3.} Turn now to the term $D_{n}.$ By the Plancherel's identity 
\begin{eqnarray*}
\left|\int_{0}^{\infty}w_{\mu,n}(v)\mathcal{F}[\bar{\nu}](-v)dv\right| & = & \left|\int_{0}^{\infty}\frac{w_{\mu,n}(v)}{(-\i v)^{s}}\left[(-\i v)^{s}\mathcal{F}[\bar{\nu}](-v)\right]dv\right|\\
 & = & \left|\int_{0}^{\infty}\frac{w_{\mu,n}(v)}{(-\i v)^{s}}\left[\mathcal{F}[\bar{\nu}^{(s)}](-v)\right]dv\right|\\
 & = & 2\pi\,\left|\int_{-\infty}^{\infty}\bar{\nu}^{(s)}(x)\overline{{\cal F}^{-1}[w_{\mu,n}(\cdot)/(-i\cdot)^{s}](x)}dx\right|\\
 & \leq & 2\pi\|\bar{\nu}^{(s)}\|_{\infty}\|\mathcal{F}^{-1}[w_{\mu,n}(\cdot)/(-i\cdot)^{s}]\|_{L^{1}}.
\end{eqnarray*}
\end{proof}
\textbf{\underline{Proof of Theorem~\ref{corub}}}
\begin{enumerate}[(i)]
\item
Suppose that $\pi\in\mathcal{P}(\beta,L)$ and $\alpha(j)\lesssim e^{-j\alpha^{*}},$
then by taking $p=c\log(n)$ for $c$ large enough, we arrive at
\begin{eqnarray*}
\mathrm{E}\left[\left|\mu_{n}-\mu\right|^{2}\right] & \lesssim & \frac{V_{n}^{-4}\log(n)}{n}\int_{0}^{V_{n}}|v|^{2\beta+2}\left|w_{\mu}(v/V_{n})\right|^{2}\, dv+V_{n}^{-2(s+2)}\\
&\lesssim& n^{-1}\log(n)V_{n}{}^{2\beta-1}+V_{n}^{-2(s+2)},\\
\mathrm{E}\left[\left|\lambda_{n}-\lambda\right|^{2}\right] & \lesssim & \frac{V_{n}^{-2}\log(n)}{n}\int_{0}^{V_{n}}|v|^{2\beta+2}\left|w_{\lambda}(v/V_{n})\right|^{2}\, dv+V_{n}^{-2(s+1)}\\&\lesssim& n^{-1}\log(n)V_{n}{}^{2\beta+1}+V_{n}^{-2(s+1)}
\end{eqnarray*}
By taking $V_{n}=n^{1/(2\beta+2s+3)},$ we get 
\[
\mathrm{E}\left[\left|\mu_{n}-\mu\right|^{2}\right]\lesssim n^{-2(s+2)/(2\beta+2s+3)}\log(n)
\]
and 
\[
\mathrm{E}\left[\left|\lambda_{n}-\lambda\right|^{2}\right]\lesssim n^{-2(s+1)/(2\beta+2s+3)}\log(n).
\]
\item 
Suppose that \textup{$\pi\in\mathcal{E}(\alpha,L),$ then by taking
$p=0,$ we get 
\begin{eqnarray*}
\mathrm{E}\left[\left|\mu_{n}-\mu\right|^{2}\right] & \lesssim & \frac{V_{n}^{-4}}{n}\left[\int_{0}^{V_{n}}\frac{\left|u^{\circ}+\i v\right|\left|w_{\mu}(v/V_{n})\right|}{\exp(-\alpha|v|)}\, dv\right]^{2}+V_{n}^{-2(s+2)}\\&\lesssim&\frac{1}{n}\exp(2\alpha V_{n})+V_{n}^{-2(s+2)},\\
\mathrm{E}\left[\left|\lambda_{n}-\lambda\right|^{2}\right] & \lesssim & \frac{V_{n}^{-2}}{n}\left[\int_{0}^{V_{n}}\frac{\left|u^{\circ}+\i v\right|\left|w_{\lambda}(v/V_{n})\right|}{\exp(-\alpha|v|)}\, dv\right]^{2}+V_{n}^{-2(s+1)}\\&\lesssim&\frac{V_{n}^{2}}{n}\exp(2\alpha V_{n})+V_{n}^{-2(s+1)}.
\end{eqnarray*}
Under the choice $V_{n}=\frac{1}{2\alpha}\log(n)-\frac{s+2}{\alpha}\log(\log(n)),$
one derives
\[
\mathrm{E}\left[\left|\mu_{n}-\mu\right|^{2}\right]\lesssim\log^{-2(s+2)}(n)
\]
and
\[
\mathrm{E}\left[\left|\lambda_{n}-\lambda\right|^{2}\right]\lesssim\log^{-2(s+1)}(n).
\]
}
\end{enumerate}
\subsection{Upper bounds for $MISE(\bar\nu_{n})$}
\label{uppmise}
\begin{prop}
 Let  the assumptions of the Proposition~\ref{propub} be fulfilled and let the kernel \(\K(\cdot)\) satisfy the assumption \eqref{assk}. Then the mean integrated squared error of the estimator \(\bar\nu_{n}(x)\) satisfies the following asymptotic relation
\begin{eqnarray*}
\textrm{MISE} (\bar\nu_{n}) & \lesssim& 
\frac{1}{n}
\int_{\R}
\frac{\left|u^{\circ}+\i v\right|^{2}\left[\K(v / V_{n})\right]^{2}}{\left|\M(u^{\circ}+\i v+1)\right|^{2}}\, dv\\
&& \hspace{1cm}
+C_{1 } V_{n}^{3} \cdot \E \left[ 
		\left( \mu_{n} - \mu \right)^{2}
	\right] 
	+
	C_{2} V_{n} \cdot \E \left[ 
		\left( \lambda_{n} - \lambda \right)^{2}
	\right] + C_{3} \frac{AL}{V_{n}^{2 s}}
\end{eqnarray*}
with some \(C_{1}, C_{2}, C_{3} > 0.\)
\end{prop}
\begin{proof}
Recall that 
\begin{eqnarray*}
	 \bar\nu_{n}(x)   &=&  \frac{1}{2 \pi} \int_{\R} e^{\i v x } \hat\F[\bar\nu] (-v) \K(-v / V_{n})   dv = \F^{-1}[\hat\F_{\bar\nu} (\cdot) \K(\cdot / V_{n})](x),
\end{eqnarray*}
and 
\begin{eqnarray*}
\hat\F[\bar\nu](-v) &=&   - Y_{n}(\uc+\ii v) + \c_{n}\cdot (\uc+\ii v)   +  \a_{n},\\
	\F[\bar\nu](-v)  &=&  - Y(\uc+\ii v) + \c\cdot (\uc+\ii v)   +  \a.
\end{eqnarray*}
By the Parsenval's identity, 
\begin{eqnarray*}
\textrm{MISE}  &=& \frac{1}{2 \pi} \E \left[ 
	\int_{\R} \left|
		\F[\bar{\nu}_{n}] (v)
		- 
		\F[\bar\nu] (v)
	\right|^{2}  dv
\right] \\
&=&
\frac{1}{2 \pi} \E \left[ 
	\int_{\R} \left|
		\hat\F[\bar\nu] (v) \K(v / V_{n})
		- 
		\F[\bar\nu] (v)
	\right|^{2}  dv
\right] 
\\ 
&=&
\frac{1}{2 \pi} \E \left[ 
	\int_{\R} \left|
		\left( 
			\hat\F[\bar\nu] (v) 
				- 
			\F[\bar\nu] (v)
		\right)
		\K(v / V_{n})
		+ 
		\left( 
			\K(v / V_{n}) -1 
		\right) 
		\F[\bar\nu] (v)
	\right|^{2}  dv
\right] \\
& \leq & 
\frac{1}{\pi} \E \left[ 
	\int_{\R} \left|
		\left( 
			\hat\F[\bar\nu] (v) 
				- 
			\F[\bar\nu] (v)
		\right)
		\K(v / V_{n})
			\right|^{2}  dv
\right]\\
&& 	\hspace{4cm}	+ 
\frac{1}{\pi} \E \left[ 
	\int_{\R} \left|
		\left( 
			\K(v / V_{n}) -1 
		\right) 
		\F[\bar\nu] (v)
	\right|^{2}  dv
\right]\\
& \leq & 
\frac{3}{\pi} \left( J_{1} +J_{2} +J_{3} \right) +\frac{1}{\pi}J_{4},
\end{eqnarray*}
where 
\begin{eqnarray*}
J_{1} &:=& \E \left[ 
	\int_{\R} \left|
			Y_{n}(\uc+\ii v) 	- 
			Y(\uc+\ii v) 
		\right|^{2}		
		\left [ 
			\K(v / V_{n})
		\right]^{2}  dv
\right], \\
J_{2}&:=& 
	A_{n} \cdot
	\E \left[ 
		\left( \mu_{n} - \mu \right)^{2}
	\right] \qquad \mbox{with} \;\; A_{n}:= \int_{\R} 
	|\uc+\ii v|^{2} \cdot \left[  \K(v / V_{n}) \right]^{2}
	dv,\\
J_{3}&:=& 	
	B_{n} \cdot
	\E \left[ 
		\left( \lambda_{n} - \lambda \right)^{2}
	\right] \qquad \mbox{with} \;\; 
	B_{n}:=\int_{\R} 
		 \left[  \K(v / V_{n}) \right]^{2}
	dv,\\
J_{4}&:=&
	\int_{\R} \left|
		\left( 
			\K(v / V_{n}) -1 
		\right) 
		\F[\bar\nu] (v)
	\right|^{2}  dv.
\end{eqnarray*}
The treatment of \(J_{1}\) is based on the observation that 
\begin{eqnarray*}
	Y_{n}(z) 	- 	Y(z)  \asymp R_{1,n} =z\,\frac{\left[\M_{n}(z)-\M(z)\right]}{\M(z+1)}.
\end{eqnarray*}
We get that
\begin{eqnarray*}
J_{1} \asymp
\int_{\R}\E \left[ 
	\left|\M_{n}(u^{\circ}+\i v)-\M(u^{\circ}+\i v)\right|^{2}
\right] 
\frac{\left|u^{\circ}+\i v\right|^{2}\left[\K(v / V_{n})\right]^{2}}{\left|\M(u^{\circ}+\i v+1)\right|^{2}}\, dv.
\end{eqnarray*}
As it was shown before,
\(
\E \left[ 
	\left|\M_{n}(u^{\circ}+\i v)-\M(u^{\circ}+\i v)\right|^{2}
\right] \lesssim n^{-1},
\)
see \eqref{bosq}-\eqref{davydov}. Therefore,  
\begin{eqnarray*}
J_{1} \lesssim
\frac{1}{n} \cdot
\int_{\R}
\frac{\left|u^{\circ}+\i v\right|^{2}\left[\K(v / V_{n})\right]^{2}}{\left|\M(u^{\circ}+\i v+1)\right|^{2}}\, dv.
\end{eqnarray*}
To complete the proof, it is sufficient to note that 
\begin{eqnarray*}
A_{n} \asymp V_{n}^{3} \cdot \int_{\R} y^{2} \left[ \K(y) \right]^{2} dy,  \qquad 
B_{n} = V_{n} \cdot \int_{\R} \left[ \K(y) \right]^{2} dy,
\end{eqnarray*}
and
\begin{eqnarray*}
J_{4} \leq 
A 
	\int_{\R} \left|
		\frac{v}{V_{n}}
	\right|^{2 s}
	\left|
		\F[\bar\nu] (v)
	\right|^{2}  dv
	\leq \frac{A L}{V_{n}^{2 s}}.
\end{eqnarray*}
\end{proof}
\textbf{\underline{Proof of Theorem~\ref{nuub}}}
\begin{enumerate}[(i)]
\item
Recall that if \(\pi \in \mathcal{P}(\beta,L), \) then 
\begin{eqnarray*}
\mathrm{E}\left[\left|\mu_{n}-\mu\right|^{2}\right] 
&\lesssim& n^{-1}\log(n)V_{n}{}^{2\beta-1}+V_{n}^{-2(s+2)},\\
\mathrm{E}\left[\left|\lambda_{n}-\lambda\right|^{2}\right] &\lesssim& n^{-1}\log(n)V_{n}^{2\beta+1}+V_{n}^{-2(s+1)},
\end{eqnarray*}
see the proof of Theorem~\ref{corub}. Taking into account that 
\(
J_{1} \lesssim
n^{-1} V_{n}^{2 \beta +3},
\)
we arrive at 
\begin{eqnarray*}
\textrm{MISE} (\bar\nu_{n}) \lesssim n^{-1} V_{n}^{3+2 \beta} +n^{-1}\log(n)V_{n}^{2\beta+1}+V_{n}^{-2(s+1)} + V_{n}^{-2s}.
\end{eqnarray*}
Choosing $V_{n}=n^{1/(2\beta+2s+3)},$ we get 
\begin{eqnarray*}
	n^{-1} V_{n}^{3+2 \beta}  = V_{n}^{-2s} \gtrsim V_{n}^{-2(s+1)}, 
\end{eqnarray*}
and therefore
\begin{eqnarray*}
\textrm{MISE} (\bar\nu_{n}) \lesssim n^{-1}\log(n)V_{n}^{2\beta+1}+ V_{n}^{-2s} \lesssim n^{-2s/(2\beta+2s+3)}.
\end{eqnarray*}
\item Similarly, we derive the upper bound for the class \(\mathcal{E}(\alpha,L).\) 
Recall that 
\begin{eqnarray*}
\mathrm{E}\left[\left|\mu_{n}-\mu\right|^{2}\right] &\lesssim&n^{-1}\exp(2\alpha V_{n})+V_{n}^{-2(s+2)},\\
\mathrm{E}\left[\left|\lambda_{n}-\lambda\right|^{2}\right] &\lesssim&n^{-1} V_{n}^{2} \exp(2\alpha V_{n})+V_{n}^{-2(s+1)}
\end{eqnarray*}
and therefore
\begin{eqnarray*}
\textrm{MISE} (\bar\nu_{n}) \lesssim n^{-1}\log(n)V_{n}^{3} e^{2 \alpha V_{n}}+ V_{n}^{-2s} \lesssim \left( \log n \right)^{- 2 s}.
\end{eqnarray*}
\end{enumerate}
\subsection{Lower bounds for $MISE$}
\label{prooflb}
\textbf{\underline{Proof of Theorem~\ref{thm3}}.}  The general idea of the proof is to apply Theorem~2.7 from \cite{Tsyb}. This theorem yields  that \eqref{lowerbound} holds, if there exists a parameterized set of L{\'e}vy triplets 
\[
\mathcal{T}_{\theta}=(1,0,\nu_{\theta}) \subset \mathcal{G}(s,R),\quad\theta\in\{0,1\}^{L}
\]  for some \(s\in \N\cup 0, R>0\), \(L>0\) and a  set of parameters $\{\theta^{(j)}, \; j=0,\ldots, M\}$ such that the following two properties hold.
\begin{enumerate}[(i)]
\item For any \(0 \leq j < k \leq M\), 
\begin{eqnarray}
\label{1}
 	\int_{\mathbb{R}}\left|\nu_{\theta^{(j)}}(x)-\nu_{\theta^{(k)}}(x)\right| ^{2} dx\geq 2\varphi_{n}.
\end{eqnarray}
\item Denote by  \(\pi_{\theta_{j}}, j=0,\ldots, M,\) the probability distribution of the exponential L{\'e}vy model $A_{j,\infty}=\int_{0}^{\infty}e^{-\xi_{j,s}}ds,$ where  $\xi_{j,s}$ is a L{\'e}vy subordinator with  triplet $\T_{\theta_{j}}.$ Then 
\begin{eqnarray}
\label{2}
	\frac{n}{M} \sum_{j=1}^{M} K \left(\pi_{\theta^{(j)}}, \pi_{\theta^{(0)}}\right) \leq \varkappa\log(M),
\end{eqnarray} 
for \(n\) large enough, where  \(K\) stands for the Kullback-Leibler divergence between models, and \(\varkappa\in (0,1/8)\).
\end{enumerate}

Below we present a detailed proof for the polynomial case.

\textbf{1. Presentation of the models.} Consider an exponential L{\'e}vy model $A_{0,\infty}=\int_{0}^{\infty}e^{-\xi_{0,s}}ds,$
where $\xi_{0,s}$ is a L{\'e}vy subordinator with a triplet $(1,0,\nu_{0})$
and $\nu_{0}(x)=abe^{-bx}$ for some $0<a\leq1,$ $0<b<1.$ It is clear that \((1,0,\nu_{0})\in \mathcal{G}(0,R)\) for some \(R>0\) and the Laplace
exponent of $\xi_{0,s}$ is given by
\begin{eqnarray*}
\phi_{0}(z) = z+\int_{0}^{\infty}(1-e^{-xz})\,\nu_{0}(x)\, dx= z\left[1+\frac{a}{z+b}\right],\qquad\Re(z)>-b,
\end{eqnarray*}
see Example~1 from Section~\ref{secsim}. For the case of general classes \(\mathcal{G}(s,R)\) with \(s>0,\) we could take a L\'evy density of the form \(\nu_0(x)=b^{1+s}x^se^{-bx} / \Gamma(s+1).\) 

Fix some \(L>0\) and let us construct now a parameterized set of L{\'e}vy triplets 
\(
\mathcal{T}_{\theta}=(1,0,\nu_{\theta}),\;\theta\in\{0,1\}^{L},
\)
with L{\'e}vy measure \(\nu_{\theta}\) defined by 
\begin{eqnarray*}
\nu_{\theta}(x):=\nu_{0}(x)+\delta\cdot\Delta_{\theta}(x), 
\end{eqnarray*}
where 
$\delta>0$ small enough,
\begin{eqnarray*}
\Delta_{\theta}(x)&:=&\bigl(g_{\theta}(x)+a(g_{\theta}\star\exp(-b\cdot)(x))\bigr)' ,\\
g_{\theta}(x) &:=&\sum_{k=L+1}^{2L} \theta_{k-L}\cos(k\gamma_{L}x)g_{0}(x),
\end{eqnarray*}
\(\theta_{k-L}\) stands for the \((k-L)\)-th component of the vector \(\theta\), $\gamma_{L}\to\infty$ as $L\to\infty$, and 
\[
g_{0}(x):= x^{-3/2}\exp(-1/x),\quad x>0.
\]

\textbf{2. Distributional properties of the models.}
In this step, we perform some technical calculations, which will be used later. It  holds
\begin{multline*}
\mathcal{L}[\Delta_{\theta}](z) = 
\int_{0}^{\infty}e^{-zx}\Delta_{\theta}(x)\, dx \\= z\left[1+\frac{a}{z+b}\right]\left[\int_{0}^{\infty}e^{-zx}g_{\theta}(x)\, dx\right]
= \phi_{0}(z)\cdot\mathcal{L}[g_{\theta}](z),
\end{multline*}
where \(\mathcal{L}[g_{\theta}](z)\) is the Laplace transform of the function \(g_{0}(\cdot)\), which is equal to
\[
\mathcal{L}[g_{\theta}](z)=\frac{1}{2}\sum_{k=L+1}^{2L}\theta_{k-L}\left[\mathcal{L}[g_{0}](z+\i \gamma_{L}k)+\mathcal{L}[g_{0}](z-\i \gamma_{L}k)\right].
\]
We see that $\int_{0}^{\infty}\Delta_{\theta}(x)\, dx=0$ and 
\[
\phi_{\theta}(z)-\phi_{0}(z)=\delta\phi_{0}(z)\mathcal{L}[g_{\theta}](z),
\]
where \(\phi_{\theta}(\cdot)\) is the  Laplace exponent of a L{\'e}vy process \(\xi_{\theta,s}\) with the L{\'e}vy triplet \(\T_{\theta}\).  Furthermore, the Laplace transform of $g_{0}$ is given by
\[
\mathcal{L}[g_{0}](u+\i v)=\sqrt{\pi}e^{-2(z_{+}+iz_{-})}
\]
with $2z_{\pm}^{2}=\sqrt{u^{2}+v^{2}}\pm u.$ 
The Mellin transform of the density $\pi_{\theta}$ corresponding
to the L{\'e}vy model $\mathcal{T}_{\theta}$ satisfies the following
functional equation 
\[
\frac{\M_{\theta}(z)}{\M_{0}(z)}=\frac{\phi_{\theta}(z)}{\phi_{0}(z)}\frac{\M_{\theta}(z+1)}{\M_{0}(z+1)}.
\]
Since
\[
\frac{\phi_{\theta}(z)}{\phi_{0}(z)}-1=\delta\mathcal{L}[g_{\theta}](z),
\]
and 
\begin{eqnarray*}
\sum_{k=1}^{\infty}\left|\mathcal{L}[g_{\theta}](z+k)\right| & \leq & C\exp\left(-\sqrt{2\mathrm{Re}(z)}-\sqrt{2\left|\mathrm{Im}(z)\right|}\right)\\
&&\hspace{3cm}\cdot\sum_{k=1}^{\infty}\sum_{j=L+1}^{2L}\exp(-\sqrt{2\gamma_{L}j}-\sqrt{2k})\\
 & \leq & C'\exp\left(-\sqrt{2\mathrm{Re}(z)}-\sqrt{2\left|\mathrm{Im}(z)\right|}\right),\quad\mathrm{Re}(z)\geq0,
\end{eqnarray*}
we derive the following infinite product representation for the ratio $\M_{\theta}(z)/\M_{0}(z)$
\[
\frac{\M_{\theta}(z)}{\M_{0}(z)}=\prod_{k=0}^{\infty}(1+\delta\mathcal{L}[g_{\theta}](z+k)).
\]
Furthermore, it can be proved that 
\[
\left|\frac{\M_{\theta}(u+\i v)}{\M_{0}(u+\i v)}-1\right|\leq c\delta\left|\mathcal{L}[g_{\theta}](u+\i v)\right|
\]
for some absolute constant $c>0.$ Note that the random variables
$A_{\theta,\infty}=\int_{0}^{\infty}e^{-\xi_{\theta,s}}ds$ with $\xi_{\theta,s}$
being a L{\'e}vy process with the triplet $\mathcal{T}_{\theta},$ satisfies
$0<A_{\theta,\infty}<1$ a.s. Moreover the density $p_{0}$ of the
r.v. $A_{0,\infty}$ has the form 
\[
\pi_{0}(x)=\frac{1}{B(b-1,a)}x^{b}(1-x)^{a-1}1_{\{0<x<1\}}
\]
 and the Mellin transform $\M_{0}(z)$ of $A_{0,\infty}$ is given
by 
\begin{equation}
\M_{0}(z)=\frac{B(z+b,a)}{B(b-1,a)},\label{eq:mellin_0}
\end{equation}
see Example~\ref{exx}.

\textbf{3. Class \(\mathcal{G}(s,R)\).} In this step, we check that constructed  models \(\mathcal{T}_{\theta^{(j)}},\) \(j=1,\ldots, M\) belong to class \(\mathcal{G}(s,R)\) with \(s=0\) and some \(R>0\). We have for any $\theta\in\{0,1\}^{L},$ 
\begin{eqnarray*}
\int_{\mathbb{R}}\left|v\right|^{2s}\left|\mathcal{F}\left[\nu_{\theta}\right](v)\right|^{2}\, dx & \leq & \int_{\mathbb{R}}\left|v\right|^{2s}\left|\mathcal{F}\left[\nu_{0}\right](v)\right|^{2}\, dv\\
&&\hspace{2cm}+\int_{\mathbb{R}}\left|v\right|^{2s}\left|\mathcal{F}\left[\nu_{\theta}\right](v)-\mathcal{F}\left[\nu_{0}\right](v)\right|^{2}\, dv\\
 & \leq & \int_{\mathbb{R}}\left|v\right|^{2s}\left|\mathcal{F}\left[\nu_{0}\right](v)\right|^{2}\, dv+\delta^{2}\int_{\mathbb{R}}\left|v\right|^{2s}\left|\mathcal{F}[\Delta_{\theta}](v)\right|^{2}\, dv.\\
\end{eqnarray*}
The inequality $\left|\phi_{0}(-\i v)\right|\leq c\cdot |v|$ for $v\in\mathbb{R}$, where \(c= 1 + a/b,\) implies
\begin{eqnarray*}
\int_{\mathbb{R}}\left|v\right|^{2s}\left|\mathcal{F}[\Delta_{\theta}](v)\right|^{2}\, dv & \leq & c \int_{\mathbb{R}}|v|^{2(s+1)}\left|\mathcal{L}[g_{\theta}](-\i v)\right|^{2}\, dv\\
 & = & \frac{c}{2}\int_{\mathbb{R}}|v|^{2(s+1)}\cdot \left|\sum_{k=L+1}^{2L}\theta_{k-L}\left(\mathcal{L}[g_{0}](-\i v+\i \gamma_{L}k) \right.\right. \\
&& \hspace{3cm} \left. \left.
 +\mathcal{L}[g_{0}](-\i v-\i \gamma_{L}k)\right)\right|^{2}\, dv\\
 & \leq & \frac{c}{2}\sum_{k=L+1}^{2L}\int_{\mathbb{R}}|v|^{2(s+1)}\left|\mathcal{L}[g_{0}](-\i v+\i \gamma_{L}k)\right|^{2}\, dv\\
 &  & \hspace{0.2cm}+\frac{c}{2}\sum_{k=L+1}^{2L}\int_{\mathbb{R}}|v|^{2(s+1)}\left|\mathcal{L}[g_{0}](-\i v-\i \gamma_{L}k)\right|^{2}\, dv\\&&\hspace{6cm}+R_{L},\\
\end{eqnarray*}
where 
\begin{eqnarray*}
R_{L} & = & 2\sum_{k\neq j}\int_{\mathbb{R}}|v|^{4}\mathcal{L}[g_{0}](-\i v-\i j\gamma_{L})\overline{\mathcal{L}[g_{0}](-\i v-\i k\gamma_{L})}\, dv\\
 &  & +2\sum_{k\neq j}\int_{\mathbb{R}}|v|^{4}\mathcal{L}[g_{0}](-\i v+\i j\gamma_{L})\overline{\mathcal{L}[g_{0}](-\i v+\i k\gamma_{L})}\, dv
\end{eqnarray*}
It holds 
\begin{eqnarray*}
\left|R_{L}\right| & \leq & CL\sum_{j=1}^{2L}\left(j\gamma_{L}\right){}^{2(s+1)}\exp(-\sqrt{2\gamma_{L}j})\\
 & \leq & CL^{2(s+1)+2}\gamma_{L}{}^{2(s+1)}\exp(-\sqrt{2\gamma_{L}})\\
 & = & o\left(L^{2(s+1)+1}\right),
\end{eqnarray*}
provided $\gamma_{L}=c\log^{2}(L)$ for large enough $c>0.$ Hence
$\int_{\mathbb{R}}\left|v\right|^{2s}\left|\mathcal{F}[\Delta_{\theta}](v)\right|^{2}\, dv$
is bounded if $\delta^{2}\gamma_{L}^{2(s+1)}L^{2s+3}=O(1).$

\textbf{4. Upper bound for the $L^{2}$-distance between elements of \(\{\nu_{\theta}\}\).}

 Fix two
vectors $\theta,\theta'\in\{0,1\}^{L}.$ We have
\begin{eqnarray*}
\int_{\mathbb{R}}\left|\nu_{\theta}(x)-\nu_{\theta'}(x)\right|^{2}\, dx & = & \frac{1}{2\pi}\delta^{2}\int_{\mathbb{R}}\left|\phi_{0}(-\i v)\mathcal{L}[g_{\theta}-g_{\theta'}](-\i v)\right|^{2}\, dv\\
 & = & \frac{1}{2\pi}\delta^{2}\sum_{k=L+1}^{2L}\left(\theta_{k-L}-\theta'_{k-L}\right)^{2}\\
 && \hspace{2cm} \cdot\int_{\mathbb{R}}\left|\phi_{0}(-\i v)\mathcal{L}[g_{0}](-\i v+\i \gamma_{L}k)\right|^{2}\, dv\\
 &  & +\frac{1}{2\pi}\delta^{2}\sum_{k=L+1}^{2L}\left(\theta_{k-L}-\theta'_{k-L}\right)^{2}\\
  && \hspace{2cm} \cdot\int_{\mathbb{R}}\left|\phi_{0}(-\i v)\mathcal{L}[g_{0}](-\i v-\i \gamma_{L}k)\right|^{2}\, dv\\
 &  & +\frac{1}{2\pi}\delta^{2}R_{L},
\end{eqnarray*}
where
\begin{eqnarray*}
R_{L} & \leq & 2\sum_{k\neq j}\int_{\mathbb{R}}\left|\phi_{0}(-\i v)\right|^{2}\mathcal{L}[g_{0}](-\i v-\i j\gamma_{L})\overline{\mathcal{L}[g_{0}](-\i v-\i k\gamma_{L})}\, dv\\
 &  & +2\sum_{k\neq j}\int_{\mathbb{R}}\left|\phi_{0}(-\i v)\right|^{2}\mathcal{L}[g_{0}](-\i v+\i j\gamma_{L})\overline{\mathcal{L}[g_{0}](-\i v+\i k\gamma_{L})}\, dv.
\end{eqnarray*}
Consider, for example, 
\begin{multline*}
\int_{\mathbb{R}}\left|\phi_{0}(-\i v)\mathcal{L}[g_{0}](-\i v+\i \gamma_{L}k)\right|^{2}\, dv = \int_{\mathbb{R}}\left|\phi_{0}(-i(v+\gamma_{L}k))\mathcal{L}[g_{0}](-\i v)\right|^{2}\, dv
\\= \int_{\mathbb{R}}\left|v+\gamma_{L}k\right|^{2}\left|1+\frac{a}{b-i(v+\gamma_{L}k)}\right|^{2} e^{-2\sqrt{2|v|}}\, dv\\=
 \gamma_{L}^{2}k^{2}\int_{\mathbb{R}}\left|1+\frac{a}{b-i(v+\gamma_{L}k)}\right|^{2}e^{-2\sqrt{2|v|}}\, dv +O(\gamma_{L}k).
\end{multline*}
So we have 
\begin{multline*}
\sum_{k=L+1}^{2L}\left(\theta_{k-L}-\theta'_{k-L}\right)^{2}\int_{\mathbb{R}}\left|\phi_{0}(-\i v)\mathcal{L}[g_{0}](-\i v+\i \gamma_{L}k)\right|^{2}\, dv \\= C\gamma_{L}^{2}\sum_{k=L+1}^{2L}\left(\theta_{k-L}-\theta'_{k-L}\right)^{2}k^{2}
+o\left(\gamma_{L}^{2}\sum_{k=L+1}^{2L}\left(\theta_{k-L}-\theta'_{k-L}\right)^{2}k^{2}\right)\\
\geq  C'\gamma_{L}^{2}L^{2}\sum_{k=1}^{L}I(\theta_{k}\neq\theta'_{k}),
\end{multline*}
as $L\to\infty$ and $\rho(\theta,\theta')=\sum_{k=1}^{L}I(\theta_{k}\neq\theta'_{k})>0.$
Analogously, 
\begin{eqnarray*}
\sum_{k=L+1}^{2L}\left(\theta_{k-L}-\theta'_{k-L}\right)^{2}\int_{\mathbb{R}}\left|\phi_{0}(-\i v)\mathcal{L}[g_{0}](-\i v-\i \gamma_{L}k)\right|^{2}\, dv & = & C''\gamma_{L}^{2}L^{2}\rho(\theta,\theta').\end{eqnarray*}
Furtheremore, one shows (see above) that 
\[
\left|R_{L}\right|=o\left(L^{3}\right).
\]

\textbf{5. Choice of  $\theta^{(0)},\ldots,\theta^{(M)}$.}

Our choice is based on the well-known Varshamov-Gilbert bound (see \cite{Tsyb}, Lemma~2.9), which implies that there are $M>2^{L/8}$
vectors $\theta^{(0)},\ldots,\theta^{(M)}\in\{0,1\}^{L}$ such that
\[
\rho(\theta^{(j)},\theta^{(k)})\geq L/8.
\]

\textbf{6. Upper bound for $K(\pi_{0},\pi_{\theta})$.}

By Parseval identity for Mellin transforms, we get 
\begin{eqnarray*}
K(\pi_{0},\pi_{\theta}) & = & \int_{0}^{1}\frac{\left|\pi_{\theta}(x)-\pi_{0}(x)\right|^{2}}{\pi_{0}(x)}\, dx\\
 & = & \int_{0}^{1}x^{-b}(1-x)^{1-a}\left|\pi_{\theta}(x)-\pi_{0}(x)\right|^{2}\, dx\\
 & \leq & \int_{0}^{1}x^{-b}\left|\pi_{\theta}(x)-\pi_{0}(x)\right|^{2}\, dx\\
 & = & \frac{1}{2\pi}\int_{-\infty}^{\infty}\left|\M_{\theta}((1-b)/2+\i v)-\M_{0}((1-b)/2+\i v)\right|^{2}\, dv\\
 & \leq & \frac{c\delta^{2}}{2\pi}\int_{-\infty}^{\infty}\left|\M_{0}((1-b)/2+\i v)\right|^{2}\left|\mathcal{L}[g_{\theta}](u+\i v)\right|^{2}\, dv.
\end{eqnarray*}
So we get
\begin{eqnarray*}
K(\pi_{0},\pi_{\theta}) & \leq & \frac{c\delta^{2}}{2\pi}\sum_{k=L+1}^{2L}\int_{\mathbb{R}}\left|\M_{0}((1-b)/2+\i v)\right|^{2}\left|\mathcal{L}[g_{0}](-\i v+\i \gamma_{L}k)\right|^{2}\, dv,\\
 &  & +\frac{c\delta^{2}}{2\pi}\sum_{k=L+1}^{2L}\int_{\mathbb{R}}\left|\M_{0}((1-b)/2+\i v)\right|^{2}\left|\mathcal{L}[g_{0}](-\i v-\i \gamma_{L}k)\right|^{2}\, dv\\
 &  & +\frac{c\delta^{2}}{2\pi}R_{L},
\end{eqnarray*}
where
\begin{eqnarray*}
R_{L} & \leq & 2\sum_{k\neq j}\int_{\mathbb{R}}\left|\M_{0}((1-b)/2+\i v)\right|^{2}\mathcal{L}[g_{0}](-\i v-\i j\gamma_{L})\overline{\mathcal{L}[g_{0}](-\i v-\i k\gamma_{L})}\, dv\\
 &  & +2\sum_{k\neq j}\int_{\mathbb{R}}\left|\M_{0}((1-b)/2+\i v)\right|^{2}\mathcal{L}[g_{0}](-\i v+\i j\gamma_{L})\overline{\mathcal{L}[g_{0}](-\i v+\i k\gamma_{L})}\, dv.
\end{eqnarray*}
The equation (\ref{eq:mellin_0}) implies that $\M_{0}(z)$ is finite
for all $z$ with $\mathrm{Re}(z)\geq0$ and
\begin{eqnarray*}
\M_{0}(u+\i v) & = & C(a,b)\frac{\Gamma(u+\i v+b)}{\Gamma(u+\i v+b+a)}\\
 & \asymp & C(a,b)e^{-a\log(u+\i v+b)}\\
 & = & C(a,b)\left((u+b)^{2}+v^{2}\right)^{-a/2}e^{i\mathrm{Arg}(u+\i v+b)},\quad u^{2}+v^{2}\to\infty.
\end{eqnarray*}
Hence 
\[
\left|\M_{0}(u+\i v)\right|\asymp C(a,b)\left((u+b)^{2}+v^{2}\right)^{-a/2},\quad u^{2}+v^{2}\to\infty
\]
and the density $\pi_{0}$ of $A_{0,\infty}$ belongs to the class
$\mathcal{P}(a,C(a,b))$ (see also Example~\ref{exx}). We have
\[
\sum_{k=L+1}^{2L}\int_{\mathbb{R}}\left|\M_{0}((1-b)/2+\i v)\right|^{2}\left|\mathcal{L}[g_{0}](-\i v+\i \gamma_{L}k)\right|^{2}\, dv=O(L^{-2a+1})
\]
and
\[
\sum_{k=L+1}^{2L}\int_{\mathbb{R}}\left|\M_{0}((1-b)/2+\i v)\right|^{2}\left|\mathcal{L}[g_{0}](-\i v-\i \gamma_{L}k)\right|^{2}\, dv=O(L^{-2a+1}).
\]
Hence
\begin{eqnarray}
\label{3}
\frac{n}{M}\sum_{m=1}^{M} K(\pi_{0},\pi_{\theta^{(m)}})\leq n\delta^{2}L^{-2a}\log(M),\quad L\to\infty
\end{eqnarray}
for large \(n\).

\textbf{7. Choice of $L$.} To complete the proof, we choose \(L\) such that the conditions \eqref{1} and \eqref{2} are fulfilled. First note that since our model belongs to the class \(\mathcal{G}(s,R)\), we can take 
$\gamma_{L}=c\log^{2}(L)$ and  $\delta^{2}= \gamma_{L}^{-2(s+1)}L^{-2s-3} \cdot O(1),$ see Step~3 of the proof for details.  Second, comparing \eqref{3} with \eqref{2}, we fix 
\(\varkappa = 
 n\delta^{2}L^{-2a}\). This leads to the choice of \(L\) as the solution of the equation 
\[L^{2a+2s+3} 
\log^{4(s+1)}(L) = n O(1)
\]
Combination of  the results from Steps~4 and 5 yields the condition \eqref{1}, because
\begin{eqnarray*}
\int_{\mathbb{R}}\left|\nu_{\theta}(x)-\nu_{\theta'}(x)\right|^{2} dx &\geq& C_{1} \delta^{2}\gamma_{L}^{2} L^{3}\\
&=& C_{2}(\log L)^{-4s} L^{-2s } \\&=&  C_{3}\left( \log L \right)^{4s \frac{-2a-1}{2a+2s+3}}
n^{-2s / (2a+2s+3)}
\end{eqnarray*}
for some \(C_{1}, C_{2}, C_{3}>0\) and \(L\) large enough.   This observation completes the proof.
\bibliographystyle{plain}
\bibliography{Panov_bibliography}

\end{document}